\documentclass[11pt,a4paper]{article}
\usepackage{amsmath, amssymb, amsthm}
\usepackage{fullpage}
\usepackage{sgame}
\usepackage{color}
\usepackage{url}
\usepackage{hyperref}

\newtheorem{theorem}{Theorem}[section]

\newtheorem{lemma}[theorem]{Lemma}

\newtheorem{obs}[theorem]{Observation}

\newcommand{\Prob}[2]{\mathbf{P}_{#1} \left( #2 \right)}
\newcommand{\Expec}[2]{\mathbf{E}_{#1} \left[ #2 \right]}

\newcommand{\x}{{\mathbf{x}}}
\renewcommand{\u}{{\mathbf{u}}}
\renewcommand{\v}{{\mathbf{v}}}
\newcommand{\0}{{\mathbf 0}}
\newcommand{\1}{{\mathbf 1}}
\newcommand{\y}{{\mathbf y}}
\newcommand{\z}{{\mathbf z}}
\newcommand{\w}{{\mathbf w}}
\newcommand{\e}{{\mathbf e}}

\newcommand{\tm}{t_\text{\rm mix}}
\newcommand{\tr}{t_\text{\rm rel}}
\newcommand{\tc}{{\tau_\text{\rm couple}}}
\newcommand{\tv}[1]{\left\|#1\right\|_{\rm TV}}
\newcommand{\pimin}{\pi_{\text{\rm min}}}
\newcommand{\hamming}{d}
\newcommand{\Bn}{B}
\newcommand{\OO}{{O}}

\renewcommand{\leq}{\leqslant}
\renewcommand{\geq}{\geqslant}
\renewcommand{\epsilon}{\varepsilon}

\newcommand{\sign}[1]{{\sf sign}\left(#1\right)}

\newcommand{\diam}{\text{\rm diam}}

\newcommand{\G}{\mathcal{G}}
\newcommand{\vset}{S}
\newcommand{\ord}{\ell}
\newcommand{\dm}{\lambda}
\newcommand{\dmin}{\delta_1}
\newcommand{\dmax}{\delta_0}
\newcommand{\deq}{\delta}
\newcommand{\pset}{\Gamma}
\newcommand{\Psub}{P^{(i,\z_{-i})}}
\newcommand{\innprod}[1]{\langle #1f,f \rangle_{\pi}}

\newcommand{\Pot}{\Phi}
\newcommand{\LV}{{\delta\Pot}}
\newcommand{\Df}{{\Delta\Pot}}
\newcommand{\GV}{\Delta\Pot}
\newcommand{\M}{\mathcal{M}}
\newcommand{\zu}{\{0,1\}}
\newcommand{\PM}{\Phi_{\max}}
\newcommand{\Pm}{\Phi_{\min}}
\newcommand{\MB}{\M^\beta}
\newcommand{\ordset}{\mathcal{L}}

\begin{document}

 \raggedbottom

\title{\bf Convergence to Equilibrium of Logit Dynamics for Strategic Games\thanks{Work partially supported by the Italian Ministry of Research
- PRIN 2008 research project COGENT. An extended abstract of this paper already appeared in \cite{AFPPP11}.}}
\author{
Vincenzo Auletta\thanks{Dipartimento di Informatica, Universit\`a di Salerno. Email: \texttt{auletta@dia.unisa.it}.}
\and
Diodato Ferraioli\thanks{LAMSADE, Universit\'e Paris-Dauphine. Email: \texttt{diodato.ferraioli@dauphine.fr}.}
\and
Francesco Pasquale\thanks{Dipartimento di Informatica, Universit\`a di Salerno. Email: \texttt{pasquale@dia.unisa.it}.}
\and
Paolo Penna\thanks{Institute of Theoretical Computer Science, ETH Zurich. Email: \texttt{paolo.penna@inf.ethz.ch}.}
\and
Giuseppe Persiano\thanks{Dipartimento di Informatica, Universit\`a di Salerno. Email: \texttt{giuper@dia.unisa.it}.}
}
\date{}
\maketitle

\thispagestyle{empty}

\begin{abstract}
We present the first general bounds on the mixing time of the Markov chain associated to the logit dynamics for wide classes of strategic games. The
logit dynamics with inverse noise $\beta$ describes the behaviour of a complex system whose individual components act selfishly and keep responding according to some
partial (“noisy”) knowledge of the system, where the capacity of the agent to know the system and compute her best move is measured by the inverse of the parameter $\beta$.

In particular, we prove nearly tight bounds for potential games and games with dominant strategies. Our results show that, for potential games,
the mixing time is upper and lower bounded by an exponential in the inverse of
the noise and in the maximum potential diﬀerence. Instead,
for games with dominant strategies, the mixing time cannot
grow arbitrarily with the inverse of the noise.

Finally, we refine our analysis for a subclass of potential games called
graphical coordination games, a class of games that have been previously
studied in Physics and, more recently, in Computer Science
in the context of diffusion of new technologies. We give evidence that the
mixing time of the logit dynamics for these games strongly depends on the
structure of the underlying graph. We prove that the mixing time of the logit dynamics
for these games can be upper bounded by a function that is exponential in the cutwidth
of the underlying graph and in the inverse of noise.
Moreover, we consider two specific and popular network topologies, the clique and the
ring. For games played on a clique we prove an almost matching lower bound on the
mixing time of the logit dynamics that is exponential in the inverse of
the noise and in the maximum potential diﬀerence, while for games played on a ring
we prove that the time of convergence of the logit dynamics to its stationary distribution is significantly shorter.
\end{abstract}

\section{Introduction}
Complex systems are often studied by looking at their dynamics and the equilibria induced by these dynamics. In this paper we concentrate on specific complex systems arising from \emph{strategic games}. Here we have a set of selfish agents or \emph{players}, each with a set of possible actions or \emph{strategies}. An agent continuously evaluates her utility or \emph{payoff}, that depends on her own strategy and on the strategies played by the other agents. A dynamics specifies the rule used by the players to update their strategies. In its most general form an equilibrium is a distribution over the set of states that has the property of being invariant with respect to the dynamics. For example, a very well-studied dynamics for strategic games is the {\em best response dynamics} whose associated equilibria are the Nash equilibria.

There are several characteristics of a dynamics and of the associated equilibrium concept that concur to make the dynamics descriptive of a system. First of all, it is desirable that the dynamics gives only one equilibrium state or, in case a system admits more than one equilibrium for a given dynamics, that the equilibria look similar. For example, this is not the case for Nash equilibria as a game can admit more than one Nash equilibrium and sometimes the equilibria have strikingly different characteristics.
In addition, the dynamics must be descriptive of the way individual agents behave. For example, the best response dynamics is well-tailored for modeling players that have a complete knowledge of the global state of the system and of their payoffs.
Finally, if a dynamics takes very long time to reach an equilibrium then the system spends most of its life outside of the equilibrium and thus knowledge gained from the study of the equilibrium is not very relevant.

In this work we study a specific \emph{noisy} best-response dynamics, the \emph{logit dynamics} (defined in \cite{blumeGEB93}) in which, at each time step, a player is randomly selected for strategy update and the update
is performed with respect to a ``noisy'' knowledge of the game and of the state of the system, that is, the strategies currently played by the players. Intuitively, ``high noise'' represents the situation where players choose their strategies ``nearly at random'' because they have a limited knowledge of the system; instead, ``low noise'' represents the situation where players ``almost surely'' play the best response; that is, they pick the strategies yielding high payoff with ``much higher'' probability. After a sufficiently large number of steps, the probability that the system is found in a specific profile remains essentially unchanged and we say that the logit dynamics has converged to a {\em stationary distribution}, that is unique and independent of the starting state. We believe that this makes the logit dynamics the elective choice of a dynamics for large and complex systems in which agents have limited knowledge. However, one more step is needed to complete the picture. How long does the logit dynamics take to converge to the stationary distribution? This is the main technical focus of this paper. Specifically, we study the \emph{mixing time} of the logit dynamics, that is, the time needed to get close to the stationary distribution. This depends on the underlying game and on the noise of the system (roughly speaking, the payoffs and how much players care about them). Since previous work has shown that the mixing time can vary a lot (from linear to exponential \cite{afppSAGT10}) it is natural to ask the following questions: (1) How do the {\em noise} level and the {\em structure} of the game affect the mixing time? (2) Can the mixing time grow \emph{arbitrarily}?

In order to answer above questions, we give general bounds on the mixing time for wide classes of games. Specifically, we prove in Section~\ref{sec::potential} that, for all \emph{potential games}, the mixing time of the logit dynamics is upper-bounded by a {\em polynomial} in the number of players and by an \emph{exponential} in the rationality level and in some structural properties of the game. However, for very small values of $\beta$ the mixing time is always polynomial in the number of players

We complement the upper bound with a lower bound showing that there exist potential games with mixing time exponential in the rationality level. Thus the mixing time can grow indefinitely in potential games as $\beta$ increases. In Section~\ref{sec:coord} we also study a special class of potential games, the {\em graphical coordination games}: we extend the result given in \cite{bkmp2005} for the Ising model; then, we give a more careful look at two extreme and well-studied cases, the clique and the ring.

Going to the second question, in Section~\ref{sec::dominant} we show that for games with dominant strategies (not necessarily potential games) the mixing time cannot exceed some \emph{absolute bound} $T$ which depends uniquely on the number of players $n$ and on the number of strategies $m$. Though $T=T(n,m)$ is of the form $\OO(m^n)$, it is independent of the rationality level and we show that, in general, such an exponential growth is the best possible.

Our results suggest that the structural properties of the game are important for the mixing time. For high $\beta$, players tend to play best response and for those games that have more than one pure Nash equilibrium (PNE) with similar potential the system is likely to remain in a PNE for a long time, whereas the stationary distribution gives each PNE approximately the same weight. This happens for (certain) potential games, whence the exponential growth of mixing time with respect to the rationality level. On the contrary, for games with dominant strategies there is a PNE (a dominant profile) with high stationary probability \emph{and} players are guaranteed to play that profile with non-vanishing probability (regardless of the rationality level).

\paragraph{Related works.} The logit dynamics was first studied by Blume~\cite{blumeGEB93} who showed that, for $2\times 2$ coordination games, the long-term behavior of the system is concentrated in the risk dominant equilibrium (see~\cite{hsMIT88}). The study of the mixing time of the logit dynamics for strategic games has been initiated in~\cite{afppSAGT10}, where, among others, bounds were given for the class of $2\times 2$ coordination games studied in \cite{blumeGEB93}. Before the work reported in \cite{afppSAGT10}, the rate of convergence was studied only for the hitting time of specific profiles; see for example the work by Asadpour and Saberi~\cite{asWINE09} who studied the hitting time of the Nash equilibrium for a class of congestion games.

Graphical coordination games are often used to model the spread of a new technology in a social network \cite{youngTR00} with the strategy of maximum potential corresponding to adopting the new technology; players prefer to choose the same technology as their neighbors and the new technology is at least as preferable as the old one.
Ellison~\cite{ellisonECO93} studied the logit dynamics for graphical coordination games on rings and showed
that some large fraction of the players will eventually choose the strategy with maximum potential.
Similar results were obtained by Peyton Young~\cite{youngTR00}
for the logit dynamics and for more general families of graphs.
Montanari and Saberi~\cite{msFOCS09} gave bounds on the hitting time
of the highest potential equilibrium for the
logit dynamics in terms of some graph theoretic properties of the
underlying interaction network.
We notice that none of \cite{blumeGEB93,ellisonECO93,youngTR00}
gave bounds on the convergence rate of the dynamics, while
Montanari and Saberi~\cite{msFOCS09} studied the convergence time of a specific configuration,
namely the {hitting} time of the highest potential equilibrium.

Our work is also strictly related to the well-studied Glauber dynamics on the Ising model
(see, for example, \cite{Mart1999} and Chapter 15 of~\cite{lpwAMS08}).
Indeed, the Ising model can be seen
as a special graphical coordination game without risk dominant equilibria,
and the Glauber dynamics on the Ising model is equivalent to the logit dynamics.
In particular, we note that Berger et al.~\cite{bkmp2005} relate the mixing time of the
Ising model to the cutwidth of the underlying graph.
Their results can be specialized to derive upper bounds on the mixing time
of graphical coordination games without risk dominant equilibria. However the bounds we present in Section~\ref{sec:coord} are tighter.

Even if the logit dynamics has attracted a lot of attention in different scientific communities,
many other promising dynamics that deal with partial or noise-corrupted knowledge of the game have been
proposed
(see, for example, the recent work of Marden et al.~\cite{myasSIAM09} and of
Mertikopoulos and Moustakas~\cite{pmIMS10} and references in~\cite{youngWINE09}).

\paragraph{Paper organization.}
We give formal definitions of logit dynamics and some of the used techniques in Section~\ref{sec::preliminaries}. The upper bounds for potential games, for games with dominant strategies, and for graphical coordination games are given in Section~\ref{sec::potential}, Section~\ref{sec::dominant}, and Section~\ref{sec:coord}, respectively.

\section{Preliminaries}\label{sec::preliminaries}
In this section
we recall basic notions about strategic games and Markov chains,
introduce the logit dynamics and describe the proof techniques that we will use for deriving upper and lower bounds on the mixing time.

\paragraph{Games.}
In a \emph{strategic game} $\G$ we have a finite set of players $\{1,\ldots, n\}$,
and each player $i$ has a finite set $S_i$ of \emph{strategies} and a \emph{utility} function
$u_i: S_1 \times \cdots \times S_n \rightarrow \mathbb R$.
A \emph{strategy profile} of $\G$ is a vector $\x=(x_1,\ldots,$ $x_n)$ with $x_i\in S_i$;
that is, in profile $\x$, player $i$ chooses strategy $x_i\in S_i$.
The {\em utility} (or {\em payoff}) of player $i$ in profile $\x$ is $u_i(\mathbf x)$.
Throughout the paper we adopt the standard game theoretic notation and write $(a,\x_{-i})$ for the
profile obtained from $\x$ by replacing the $i$-th entry with $a$;
i.e.,  $(a,\x_{-i}) = (x_1,\ldots,x_{i-1}, a, x_{i+1},\ldots,x_n)$.
We also denote with $S$ the set $S:=S_1 \times \cdots \times S_n$
of all strategy profiles of the game.

A strategic game $\G$ is a \emph{potential game} if there
exists a {\em potential} function $\Pot \colon S \rightarrow \mathbb{R}$
such that for every player $i$, every pair of strategies $a, b \in S_i$,
and every profile $\mathbf{x} \in S$, it holds that
\begin{equation}
\label{eq:defpotential}
u_i(a, \mathbf{x}_{-i}) - u_i(b, \mathbf{x}_{-i}) = \Pot(b,\mathbf{x}_{-i}) - \Pot(a,\mathbf{x}_{-i})\,.
\end{equation}

\paragraph{Markov chains.}
A sequence of random variables $(X_0, X_1, \ldots)$ is a \emph{Markov chain} $\M$ with \emph{state space} $\Omega$ and \emph{transition matrix} $P$ if for all $x, y \in \Omega$, all $t \geq 1$, and all events $H_{t-1} = \bigcap_{s=0}^{t-1} \{X_s = x_s\}$ satisfying $\Prob{}{H_{t-1} \cap \{X_t = x\}} > 0$, we have
$$
\Prob{}{X_{t+1} = y \mid H_{t-1} \cup \{X_t = x\}} = \Prob{}{X_{t+1} = y \mid X_t = x} = P(x,y)\,.
$$
We denote with $\Prob{x}{\cdot}$ and $\Expec{x}{\cdot}$ the probability and the expectation conditioned on the starting state of the Markov chain being $x$, i.e., on the event $\{X_0 = x\}$. The \emph{$t$-step transition matrix} $P^t$ sets $P^t(x,y) = \Prob{x}{X^t = y}$.

A Markov chain $\M$ is \emph{irreducible} if for any two states $x,y \in \Omega$ there exists an integer $t = t(x,y)$ such that $P^t(x, y) > 0$;
i.e., it is possible to reach any state from any other one. The \emph{period} of an irreducible Markov chain is the greatest common divisor of $\{t \geq 1 \colon \exists x \text{ such that } P^t(x,x) > 0\}$. If the period of a Markov chain is greater than 1, then the chain is called \emph{periodic}, otherwise \emph{aperiodic}. If a Markov chain is finite (i.e., the space state $\Omega$ is a finite set), irreducible and aperiodic then the chain is \emph{ergodic}: for an ergodic chain there is an integer $r$ such that, for all $x,y \in \Omega$, $P^r(x,y) > 0$.

It is a classical result that if $\M$ is ergodic there it converges to an unique {\em stationary distribution}. That is, there exists a distribution $\pi$ on $\Omega$ such that $\pi P=\pi$ and, for every initial profile $x \in \Omega$, the distribution $P^t(x,\cdot)$ of the chain at time $t$ converges to $\pi$ as $t$ goes to infinity.

A Markov chain $\M$ is \emph{reversible} if for all $x,y\in\Omega$, it holds that
$$
 \pi(x) P(x,y)=\pi(y) P(y,x)\,.
$$
The probability distribution $Q(x,y)=\pi(x)P(x,y)$ over $\Omega \times \Omega$ is sometimes called \emph{edge stationary distribution}.

The \emph{mixing time} of a Markov chain is the time needed for $P^t(x,\cdot)$ to be close to $\pi$ for every initial state $x$. Specifically,
$$
\tm(\varepsilon) := \min \left\{t \in \mathbb{N} \colon \tv{P^t(x, \cdot) - \pi} \leqslant \varepsilon \mbox{ for all } x \in \Omega \right\},
$$
where
$\tv{P^t(x,\cdot)-\pi} = \frac{1}{2} \sum_{y\in \Omega} \left|P^t(x,y)-\pi(y)\right|$
is the \emph{total variation distance} among the probability distributions $P^t$ and $\pi$.
We will use the standard convention of setting $\tm=\tm(1/4)$ and observe
that $\tm(\varepsilon)\leqslant\tm(1/4)\cdot\log{1/\varepsilon}$.

\paragraph{Logit dynamics.}
The {logit dynamics} (see \cite{blumeGEB93}) for a strategic game $\G$ runs as follows: at every time step a player $i$ is selected uniformly at random and she updates her strategy to $y\in S_i$
with probability $\sigma_i(y\mid\x)$ defined  as
\begin{equation}\label{eq:updateprob}
\sigma_i(y \mid \x):= \frac{1}{T_i(\mathbf{x})} \, e^{\beta u_i(y,\x_{-i})}\,,
\end{equation}
where $\x\in S$ is the current strategy profile, $T_i(\x)=\sum_{z \in S_i} e^{\beta u_i(z, \x_{-i})}$ is the normalizing factor, and parameter $\beta \geqslant 0$ is the \emph{inverse noise} (or \emph{inverse temperature}).

The logit dynamics for $\G$ defines a Markov chain $\MB(\G) = \{ X_t \colon t \in \mathbb{N} \}$ with state space $S = S_1 \times \cdots \times S_n$ and transition probabilities
\begin{equation}
\label{eq:transmatrix}
P(\x,\y) = \frac{1}{n}\cdot
\begin{cases}
  \sigma_i(y_i\mid\x), & \text{if } \x\neq\y\text{ and }\x_{-i}=\y_{-i};\\
  \sum_{i=1}^n \sigma_i(y_i\mid\x), & \mbox{if } \x=\y; \\
  0, & \mbox{otherwise}.
\end{cases}
\end{equation}
We will find convenient to identify the logit dynamics for $\G$ with the Markov chain $\MB(\G)$.
It is easy to see that $\MB(\G)$ is ergodic.
Therefore, there exists a unique \emph{stationary distribution} $\pi$ and, for every initial profile $\x$, the distribution $P^t(\x,\cdot)$ of the chain at time $t$ converges to $\pi$ as $t$ goes to infinity.

It is easy to see that, if $\G$ is a potential game with potential function $\Pot$, then $\MB(\G)$ is reversible and the stationary distribution is the Gibbs measure
\begin{equation}\label{eq:Gibbs}
\pi(\x) = \frac{1}{Z} e^{\beta \Pot(\x)}
\end{equation}
where $Z=\sum_{\y \in S} e^{\beta \Pot(\y)}$ is the normalizing constant (also called the {\em partition function}). We will write $Z_\beta$ and $\pi_\beta$ when we will need to stress the dependence on the inverse noise $\beta$.

\paragraph{Further notation.}
We use bold symbols for vectors.
We denote by $\hamming(\x,\y)$ the \emph{Hamming distance} between $\x,\y$, that is the number of coordinates in which these two vectors differ. Given a set $A$ of vectors, the \emph{Hamming graph on $A$} is the graph with vertex set $A$ and an edge between $\x$ and $\y$ if and only if $d(\x,\y) = 1$. For each edge $\e = (\x, \y)$  of the Hamming graph we say that $\e$ goes along the dimension $i$ of the Hamming graph if $\x$ and $\y$ differ on the position $i$ and denote its dimension by $\dm(\e)$.

\subsection{Proof techniques}
\label{subsec::spectral}
In order to derive our bounds on the mixing time of the logit dynamics we will use the following well-established techniques: \emph{Markov chain coupling} and \emph{spectral methods} for the upper bounds and the \emph{bottleneck ratio theorem} for the lower bounds.
In the rest of this section we will state the the theorems we will use
for obtaining our bounds. For a more detailed description we refer the reader to~\cite{lpwAMS08}.

\paragraph{Markov chain coupling.}
A {\em coupling}\index{coupling} of two probability distributions $\mu$ and $\nu$  on $\Omega$ is a pair of random variables $(X,Y)$ defined on $\Omega\times\Omega$ such that the marginal distribution of $X$ is $\mu$ and the marginal distribution of $Y$ is $\nu$. A {\em coupling of a Markov chain} $\mathcal{M}$ with transition matrix $P$ is a process $(X_t,Y_t)_{t=0}^\infty$ with the property that both $X_t$ and $Y_t$ are Markov chains with transition matrix $P$. When the two coupled chains start at $(X_0,Y_0) = (x,y)$, we write $\Prob{x,y}{\cdot}$ and $\Expec{x,y}{\cdot}$ for the probability and the expectation on the space where the two coupled chains are both defined.
We denote by $\tc$ the first time the two chains meet; that is,
$$
\tc=\min\{t: X_t=Y_t\}\,.
$$
We will consider only couplings of Markov chains with the property that
for $s\geq\tc$, it holds that $X_s=Y_s$.
The following theorem establishes that the total variation distance between $P^t(x,\cdot)$ and $P^t(y,\cdot)$, the distributions of the chain at time $t$ starting at $x$ and $y$ respectively, is upper bounded by the probability that two coupled chains starting at $x$ and $y$ have not yet meet at time $t$ (see, for example, Theorem~5.2 in \cite{lpwAMS08}).

\begin{theorem}[Coupling]
\label{thm:coupling}
Let $\mathcal{M}$ be a Markov chain with finite state space $\Omega$ and
transition matrix $P$. For each pair of states $x,y\in\Omega$ consider a coupling $(X_t,Y_t)$ of $\mathcal{M}$ with starting states $X_0=x$ and $Y_0=y$.
Then
$$
\tv{P^t(x,\cdot) - P^t(y,\cdot)}\leq \Prob{x,y}{\tc>t}.
$$
\end{theorem}

Since for an ergodic chain with transition matrix $P$ and stationary distribution $\pi$,
for every $x \in \Omega$ it holds that $\tv{P^t(x,\cdot) - \pi} \leqslant \max_{y \in \Omega} \tv{P^t(x,\cdot) - P^t(y,\cdot)}$, the above theorem gives a useful tool for upper bounding the mixing time.

Sometimes it is difficult to specify a coupling and to analyze the coupling time $\tc$ for every pair of starting states $(x,y)$. The \emph{path coupling}\index{path coupling} theorem says that it is sufficient to define a connected graph over the Markov chain state space and consider only couplings of pairs of Markov chains starting from \emph{adjacent} states. An upper bound on the mixing time can then be obtained if each one of those couplings contracts its distance on average. More precisely, consider a Markov chain $\mathcal{M}$ with state space $\Omega$ and transition matrix $P$; let $G=(\Omega,E)$ be a connected graph and let $w:E\rightarrow\mathbb{R}$ be  a function assigning weights to the edges such that $w(e)\geq 1$ for every edge $e\in E$; for every $x,y\in\Omega$, we denote by $\rho(x,y)$ the weight of the shortest path in $G$ between $x$ and $y$. The following theorem holds.

\begin{theorem}[Path Coupling~\cite{BubleyDyer97}]
\label{theorem:pathcoupling}
Suppose that for every edge $(x,y)\in E$ a coupling $(X_t,Y_t)$ of $\mathcal{M}$ with $X_0=x$ and $Y_0=y$ exists such that $\Expec{x,y}{\rho(X_1,Y_1)} \leqslant e^{-\alpha}\cdot w(\{x,y\})$
for some $\alpha > 0$.
Then
$$
\tm(\varepsilon) \leqslant \frac{\log(\diam(G)) + \log(1/\varepsilon)}{\alpha}
$$
where $\diam(G)$ is the (weighted) diameter of $G$.
\end{theorem}

\paragraph{Spectral methods.}
Let $P$ be the transition matrix of a Markov chain with finite state space $\Omega$
and let us label the eigenvalues of $P$ in non-increasing order
$$
\lambda_1\geqslant \lambda_2 \geqslant \dots \geqslant \lambda_{|\Omega|}.
$$
It is well-known (see, for example, Lemma~12.1 in~\cite{lpwAMS08}) that $\lambda_1 = 1$ and,
if $P$ is ergodic then $\lambda_2<1$ and $\lambda_{|\Omega|}>-1$.
We denote by $\lambda^\star$ the largest absolute value among eigenvalues other than $\lambda_1$.
The {\em relaxation time}\index{relaxation time} $\tr$ of an ergodic Markov chain $\M$
is defined as
$$
\tr = {\frac{1}{1-\lambda^\star}} = \max \left\{\frac{1}{1 - \lambda_2}, \frac{1}{1 + \lambda_{|\Omega|}}\right\}.
$$
The relaxation time is related to the mixing time by the following theorem (see, for example, Theorems 12.3 and 12.4 in \cite{lpwAMS08}).

\begin{theorem}[Relaxation time]\label{theorem:relaxation}
Let $P$ be the transition matrix of a reversible ergodic Markov chain with state space
$\Omega$ and stationary distribution $\pi$. Then
$$
(\tr-1)\cdot\log\left({\frac{1}{2\epsilon}}\right)
\leq \tm(\epsilon)\leq
\tr\cdot\log\left({\frac{1}{\epsilon\pimin}}\right),
$$
where
$\pimin=\min_{x\in\Omega} \pi(x)$.
\end{theorem}

The following theorem shows how to use the coupling technique to obtain an upper bound on the relaxation time as well (see Theorem~13.1 in \cite{lpwAMS08}).
\begin{theorem}[\cite{chen95}]
\label{th:chen}
Let $P$ the transition matrix of a Markov chain $\M$ with state space $\Omega$ and let $\rho$ be a metrics on $\Omega$.
Suppose there exists a constant $\theta < 1$ such that for each
$x,y\in\Omega$ there exists a coupling $(X,Y)$ of $P(x,\cdot)$ and $P(y,\cdot)$ such that
$$
\Expec{x,y}{\rho(X,Y)} \leq \theta \cdot \rho(x,y)\,.
$$
Then the relaxation time of $M$ is $\tr \leq \frac{1}{1-\theta}$.%
\end{theorem}

Consider two ergodic reversible Markov chains $\M$ and $\hat\M$ over the same state set $\Omega$.
We denote by $\pi$ and $\hat\pi$ and by $Q$ and $\hat Q$ the respective stationary distributions
and edge stationary distributions.
We define the set of the {\em edges of $\M$} as the set of pairs $x,y\in\Omega$ such that $P(x,y)>0$.
For $x,y\in\Omega$,
an {\em $\M$-path} is a sequence $\Gamma_{x,y}=(e_1,\ldots,e_k)$ of {edges} of $\M$
such that $e_i=(x_{i-1},x_i)$, for $i=1,\ldots,m$, and $x_0=x$ and $x_m=y$.
The length $k$ of path $\Gamma_{x,y}$ is denoted by $|\Gamma_{x,y}|$.
Let $\Gamma$ be a set of paths $\Gamma=\{\Gamma_{x,y}\}$, one for each edge $(x,y)$ of $\hat\M$.
We define the {\em congestion ratio} $\alpha$ of $\Gamma$ as
$$\alpha=
    \max_{e\in E}\left(
        \frac{1}{Q(e)}\sum_{\stackrel{x,y}{e\in\Gamma_{x,y}}} \hat Q(x,y) |\Gamma_{x,y}|
    \right)
$$
where $E$ is the set of edges of $\M$.

The following theorem relates $\lambda_2$ and $\hat\lambda_2$,
the second eigenvalues of of $\M$ and $\hat\M$, respectively.
\begin{theorem}[Path Comparison Theorem]
\label{pathcompth}
Let $\M$ and $\hat \M$ be two ergodic reversible Markov chains over the same state space $\Omega$
and let $\lambda_2$ and $\hat\lambda_2$ be their respective second eigenvalues.
If
there exists a set $\Gamma$ of $\M$-paths, containing one path for each edge of $\hat\M$,
with congestion ratio $\alpha$,
then
$$\frac{1}{1-\lambda_2}\leq\alpha\cdot\gamma\cdot\frac{1}{1-\hat\lambda_2}\,,$$
where
$\gamma = \max_{x\in\Omega} \pi(x)/\hat{\pi}(x)$.
\end{theorem}
The following theorem is a special case of the  Path Comparison Theorem
obtained by considering a Markov chain $\M$ with stationary distribution $\pi$
and a Markov chain $\hat\M$ with transition probability $\hat P(x,y)=\pi(y)$.
\begin{theorem}[Canonical paths~\cite{js89}]
\label{comp_lemma}
Let $\M$ be a reversible ergodic Markov chain over state space
$\Omega$ with transition matrix $P$ and stationary distribution $\pi$.
For each pair of profiles $x,y \in \Omega$, let $\Gamma_{x,y}$ be an $\M$-path.
The \emph{congestion} $\rho$ of the set of paths is defined as
$$
\rho = \max_{e\in E} \left(\frac{1}{Q(e)} \sum_{\begin{subarray}{c}x,y \colon\\ e \in \Gamma_{x,y}\end{subarray}} \pi(x)\pi(y)|\Gamma_{x,y}|\right).
$$
Then it holds that $\frac{1}{1 - \lambda_2} \leq \rho$.
\end{theorem}

\paragraph{Bottleneck ratio.}
\label{subsec::bottleneck}
Let $\mathcal{M} = \{ X_t \colon t \in \mathbb{N} \}$ be an  ergodic Markov chain with
finite state space $\Omega$, transition matrix $P$, and stationary distribution $\pi$. For a set of states $R \subseteq \Omega$, the \emph{bottleneck ratio}\index{bottleneck ratio} at $R$ is defined as
$$
\Bn(R) = \frac{Q(R,\overline{R})}{\pi(R)}\,,
$$
where $Q(R, \overline{R}) = \sum_{x \in R, \, y \in \overline{R}} Q(x,y)$.
The following theorem states that, for every $R$ with $\pi(R) \leqslant 1/2$, the mixing time is larger than the reciprocal of the bottleneck ratio at $R$, up to a constant factor (see, for example, Theorem~7.3 in \cite{lpwAMS08}).

\begin{theorem}[Bottleneck ratio]\label{theorem:bottleneck}
Let $\mathcal{M} = \{ X_t \,:\, t \in \mathbb{N} \}$ be an ergodic Markov
chain with finite state space $\Omega$,
transition matrix $P$, and stationary distribution $\pi$.
Let $R \subseteq \Omega$ be any set with $\pi(R) \leqslant 1/2$.
Then the mixing time is
$$
\tm(\varepsilon) \geq \frac{1-2\epsilon}{2 \Bn(R)}\,.
$$
\end{theorem}

\section{Potential games}\label{sec::potential}
In this section we give bounds on the mixing time of the Markov chains of the logit dynamics
with inverse noise $\beta$ of potential games.

We start by proving a property of the eigenvalues of the transition matrix of the Markov chain
of the logit dynamics with inverse noise $\beta$ of a potential game.
This result, that we think to be interesting by itself, shows that the relaxation time only depends
on the second eigenvalue of the transition matrix.
We then give bounds for three different ranges of $\beta$.
Our first bound holds for all values of $\beta$.
Then we show a slightly better bound for low values of $\beta$.
Finally, we give a more precise bound on the mixing time for high values of $\beta$.

\subsection{Eigenvalues of the Logit Dynamics}
We will show that the second eigenvalue of the transition matrix of the Markov chain of the logit dynamics with inverse noise $\beta$ for a potential game is always larger in absolute value than the last eigenvalue.
Hence, for these games, $\tr = \frac{1}{1-\lambda_2}$.

\begin{theorem}
 \label{thm:diod_conj}
Let $\G$ be an $n$-player potential game with profile space $S$ and let $P$ be the transition matrix of the Markov chain of the logit dynamics with inverse noise $\beta$ for $\G$. Let $1 = \lambda_1 \geq \lambda_2 \geq \ldots \geq \lambda_{|S|}$ be the eigenvalues of $P$. Then $\lambda_2 \geq \left|\lambda_{|S|}\right|$.
\end{theorem}
\begin{proof}
To prove the theorem we will show that all eigenvalues are non-negative
which implies that $\lambda_{|S|} \geq 0$. Then, we have that
$ \left|\lambda_{|S|}\right| = \lambda_{|S|} \leq \lambda_2$.

Assume for sake of contradiction that there exists an eigenvalue $\lambda<0$ of $P$ and let $f$ be an eigenfunction of $\lambda$. By definition, $f \neq \0$.
Since $\lambda < 0$, then for every profile $\x \in S$ such that $f(\x) \neq 0$,
we have $\sign{(Pf)(\x)}=\sign{\lambda f(\x)} \neq \sign{f(\x)}$
and thus
 $$
  \langle Pf, f \rangle_{\pi} := \sum_{\x \in S} \pi(\x) (Pf)(\x) f(\x) < 0\,,
 $$
where $\pi$ is the stationary distribution of the Markov chain of the logit dynamics with inverse noise $\beta$ for $\G$.

For every player $i$ and for every strategy sub-profile $\z_{-i}$, we consider the \emph{single-player matrix} $\Psub$ defined as
\[
 \Psub(\x,\y) := \frac{1}{T_i(\z_{-i})}\left\{\begin{array}{ll}
e^{\beta u_i(\y)}, & \mbox{if $\x_{-i}=\y_{-i}=\z_{-i}$}\,; \\
0, & \mbox{otherwise.}
\end{array}
\right.
\]
The transition matrix $P$ is the sum of all such ``single-player'' matrices:
\[
 P = \frac{1}{n}\sum_i \sum_{\z_{-i}} \Psub\,.
\]
Let us define
\[
 S_{i,\z_{-i}} := \left\{\x \mid \x = (s_i,\z_{-i}) \mbox{ and } s_i \in S_i\right\}.
\]
For any $\x, \y \in S_{i,\z_{-i}}$ we have that
\[
 \frac{e^{\beta u_i(\y)}}{e^{\beta u_i(\x)}} = e^{-\beta(\Phi(\y)-\Phi(\x))}
\ \ \mbox{ which implies } \ \  \frac{e^{\beta u_i(\x)}}{e^{-\beta \Phi(\x)}}=\frac{e^{\beta u_i(\y)}}{e^{-\beta \Phi(\y)}}\,.
\]
Thus, the ratio $r_{i,\z_{-i}}:= \frac{e^{\beta u_i(\z)}}{e^{-\beta\Phi(\z)}}$ is constant over all $\z \in S_{i,\z_{-i}}$. Hence, whenever $\Psub(\x,\y)$ is not zero, it does not depend on $\x$: indeed,
$$
 \Psub(\x,\y) = \frac{e^{\beta u_i(\y)}}{T_i(\z_{-i})} = \frac{Z \cdot r_{i,\z_{-i}}}{T_i(\z_{-i})} \pi(\y)\,.
$$
Setting $C_{i, \z_{i}} = \frac{Z \cdot r_{i,\z_{-i}}}{T_i(\z_{-i})}$, we obtain that
\[
 \innprod{\Psub} = C_{i,\z_{-i}} \sum_{\x \in S_{i, \z_{-i}} } \sum_{\y\in S_{i, \z_{-i}} } \pi(\x) \pi(\y) f(\x) f(\y) = C_{i,\z_{-i}}\left(\sum_{\x \in S_{i, \z_{-i}} } \pi(\x) f(\x) \right)^2 \geq 0 \,.
\]
From the linearity of the inner product, it follows that
\[
 \innprod{P} = \frac{1}{n}\sum_{i} \sum_{\z_{-i}}  \innprod{\Psub} \geq 0\,,
\]
contradicting the hypothesis.
\end{proof}

\subsection{\texorpdfstring{Bounds for all $\beta$}{Bound for all beta}}
In this section,
we give an upper bound on the mixing time of the Markov chain of the logit dynamics with inverse noise $\beta$
for potential games.
The upper bound is expressed in terms of
the maximum global variation $\GV$ of the potential function $\Pot$:
$$
\GV := \max \{ \Pot(\x) - \Pot(\y) \} = \PM - \Pm
$$
where $\PM$ and $\Pm$ are the maximum and the minimum value of $\Pot$, respectively.
The upper bound holds for every value of the inverse noise $\beta$.
We shall also provide
examples of games whose logit dynamics has a mixing time close to
the given upper bound.

\paragraph{The upper bound.}
Let $\G$ be a strategic game with
profile space $S$
and let $\MB$ be the Markov chain of the logit dynamics for $\G$ with
inverse noise $\beta$.
By Theorem~\ref{theorem:relaxation}, to get an upper bound on the mixing time
of $\MB$ it suffices to give an upper bound on the relaxation time $\tr^\beta$
of $\MB$.
We obtain an upper bound on the relaxation time of $\MB$
by comparing it with $\M^0$
(the Markov chain of the logit dynamics with inverse noise $\beta=0$).
$\M^0$ is a random walk on a regular graph with $|S|$ vertices and
its stationary distribution $\pi_0$ is the uniform distribution $\pi_0(\x)=1/|S|$.
An upper bound on its relaxation time is given by the following lemma.
\begin{lemma}\label{le:hypercube}
Let $\G$ be an $n$-player game.
The relaxation time of
the Markov chain $\M^0$ of the logit dynamics for $\G$
with inverse noise $\beta=0$ is $\tr^0\leq n$.
\end{lemma}
\begin{proof}
Let $P$ be the transition matrix of $\M^0$.
For any two profiles $\x,\y$ of $\G$ we define
the following coupling $(X,Y)$ of distributions $P(\x,\cdot)$ and $P(\y,\cdot)$:
pick $i\in [n]$ and $s\in S_i$ uniformly at random and update
$X$ and $Y$ by setting $X_i=s$ and $Y_i=s$.
Notice that when the coupling picks a player $j$ on which $\x$ and $\y$ differ,
the Hamming distance $\hamming(X,Y)$ of $X$ and $Y$ decreases by one;
otherwise, it remains the same. Thus,
$$\Expec{\x,\y}{\hamming(X,Y)} =
\frac{\hamming(\x,\y)}{n} \Bigl(\hamming(\x,\y) - 1\Bigr) +
\left(1-\frac{\hamming(\x,\y)}{n}\right) \hamming(\x,\y) =
    \left(1-\frac{1}{n}\right) \hamming(\x,\y)\,.$$
The lemma follows by applying Theorem~\ref{th:chen} with
$\theta=\left(1-\frac{1}{n}\right)$.
\end{proof}

The following lemma is the main technical result of this section.
\begin{lemma}
\label{th:potential:relaxation-time}
Let $\G$ be a $n$-player potential game where each player has at most $m$ strategies.
Let $\GV$ be the maximum global variation of the potential $\Pot$ of $\G$.
Then the relaxation time of the Markov chain $\MB$ of the logit dynamics for $\G$
with inverse noise $\beta$ is
$$\tr^\beta\leq 2mn \cdot e^{\beta \GV}.$$
\end{lemma}
\begin{proof}
Fix $\beta>0$ and
denote by $\pi_0$ and $\pi_\beta$
the stationary distributions of $\M^0$ and $\MB$.
For all $\x\in S$,
$\pi_0(\x)=\frac{1}{|S|}$ and
$$
\pi_\beta(\x) = \frac{e^{-\beta \Pot(\x)}}{Z_\beta} \leqslant
            \frac{e^{-\beta \Pm}}{Z_\beta}\,,
$$
where $Z_\beta = \sum_{\y \in S} e^{-\beta \Pot(\y)}$ is the partition function
of $\MB$.
Therefore, for all $\x$,
$$\frac{\pi_\beta(\x)}{\pi_0(\x)}\leq \gamma\quad\mbox{\rm with }\ \gamma=\frac{|S|}{Z_\beta}\cdot e^{-\beta\Pm}.$$
Next we define an assignment of $\MB$-paths to the edges of $\M^0$.
The edges $(\u,\v)$ of $\M^0$ (and of $\MB$) consist of
two profiles $\u$ and $\v$ that differ only in the strategy of one player.
With each edge of $\M^0$,
we associate a path consisting of at most two {admissible} edges of $\MB$.
We say that edge $(\u,\z)$ of $\MB$ is an \emph{admissible edge} if
one of $\u$ and $\z$ minimizes the potential function $\Pot$
over all profiles that are in the intersection of the neighborhoods of
$\u$ and $\z$.
Notice that, since $\u$ and $\z$ differ exactly in one component, say the $j$-th,
$\u_{-j}=\z_{-j}$ and
the intersection of the two neighborhoods consists of all
profiles of the form $(\u_{-j},s)$ with $s\in S_j$.

If edge $(\u,\v)$ of $M^0$ is admissible then it is associated with the path  $\Gamma_{\u,\v}$
consisting solely of the edge $(\u,\v)$ of $\MB$.
Otherwise, let $\z$ be the profile of minimum potential in the
interesection of the neighborhoods of $\u$ and $\v$.
Path $\Gamma_{\u,\v}$ of edge $(\u,\v)$ of $\M^0$
then consists of the edges $(\u,\z)$ and $(\z,\v)$ of $\M^\beta$.

Let us now upper bound the congestion ratio $\alpha$ of the assignment.
By the reversibility of $\MB$, we have that
$Q_\beta(\u,\z)=Q_\beta(\z,\u)$ and therefore,
without loss of generality, we compute $Q_\beta(\e)$ for $\e=(\u,\z)$
for which $\z$ is the minimum potential profile from the
common neighborhood of $\u$ and $\z$.
For such an edge $\e$, by denoting by $i$ the one component in which
$\u$ and $\z$ differ, we have
$$
Q_\beta(\e)=
\frac{e^{-\beta\Pot(\u)}}{Z_\beta} \frac{e^{-\beta\Pot(\z)}}{n\cdot\sum_{s\in S_i} e^{-\beta\Pot(\z_{-i},s)}}
\geq
\frac{e^{-\beta\PM}}{Z_\beta}
            \frac{1}{n\cdot\sum_{s\in S_i} e^{\beta(\Pot(\z)-\Pot(\z_{-i},s))}}.
$$
Observe that $\Pot(\z)\leq\Pot(\z_{-i},s)$ for all $s\in S_i$
and thus each of the $|S_i|$ terms
in the sum at the denominator is at most $1$. Hence
$$
Q_\beta(\e)\geq
\frac{e^{-\beta\PM}}{Z_\beta}
            \frac{1}{n\cdot|S_i|}.
$$
Now, edge $\e=(\u,\z)$ is used on all paths
for edges $(\u,\v)$ of $\M^0$
for which $\u$ and $\v$ differ only in the $i$-th component.
Thus we have
$Q_0(\u,\v)=1/|S|\cdot 1/(n\cdot{|S_i|})$.
Finally, since there are exactly $|S_i|$ such edges of $\M^0$
and their corresponding paths in $\MB$ have length at most $2$,
we can write the congestion ratio for edge $\e=(\u,\z)$ of $\MB$ as
$$
        \frac{1}{Q_\beta(\e)}\sum_{\stackrel{\u,\v}{\e\in\Gamma_{\u,\v}}} Q_0(\u,\v) |\Gamma_{\u,\v}|
    \leq
\frac{Z_\beta}{e^{-\beta\PM}}
            \cdot {n\cdot|S_i|}\cdot
\sum_{\stackrel{\u,\y}{\e\in\Gamma_{\u,\y}}} \frac{2}{n\cdot |S_i|\cdot |S|}\leq
2\cdot \frac{Z_\beta}{|S|}\cdot e^{\beta\PM}\cdot |S_i|
$$
By applying
the Path Comparison Theorem (Theorem~\ref{pathcompth})  with
$$
\alpha = \frac{2m\cdot Z_\beta}{|S|}\cdot{e^{\beta \PM}}
\qquad \mbox{ and } \qquad
\gamma = \frac{|S|}{Z_\beta}\cdot e^{-\beta \Pm},
$$
and since, by Theorem~\ref{thm:diod_conj} $\tr^\beta=1/(1-\lambda_{|S|})$,
we obtain that
$$\tr^\beta\leq 2m\cdot\tr^0\cdot e^{\beta\GV}.$$
The lemma then follows from Lemma~\ref{le:hypercube}.
\end{proof}

We now give an upper bound on the mixing time of $\MB$, for all $\beta\geq 0$.
\begin{theorem}\label{thm:potential:ub}
Let $\G$ be a $n$-player potential game where each player has at most $m$ strategies.
Let $\GV$ be the maximum global variation of the potential $\Pot$ of $\G$.
Then the mixing time $\tm^\beta$
of the Markov chain $\MB$ of the logit dynamics for $\G$
with inverse noise $\beta$ is
$$\tm^\beta(\varepsilon)\leq 2mn\cdot e^{\beta \Delta\Pot}
        \left(\log \frac{1}{\varepsilon}+\beta\Delta\Pot + n\log m \right).
$$
\end{theorem}
\begin{proof}
Use Theorem~\ref{theorem:relaxation},
the upper bound on the relaxation time provided by Lemma~\ref{th:potential:relaxation-time}
and the fact that, for all $\x$,
$\pi_\beta(\x) \geqslant 1/\left(e^{\beta\Delta\Pot} |S|\right)$.
\end{proof}
\smallskip\noindent{\it Remark.}
Notice that we can write the upper bound provided by
Theorem~\ref{thm:potential:ub} as
$$\tm^\beta \leq e^{\beta\Delta\Pot(1+o(1))}
$$
where the asymptotic is in $\beta$.

\paragraph{The lower bound.}
In this section we provide a lower bound on the mixing time of the Markov chain
of the logit dynamics with inverse noise $\beta$ for potential games.
We start by observing that it is easy to find a potential $\Pot$ for which
the Markov chain of the logit dynamics with inverse noise $\beta$ has mixing
time $\Omega(e^{\beta \GV})$. Specifically, consider a potential $\Pot$ on
$\zu^n$ such that $\Pot(\0)=\Pot(\1)=0$ and, for some $L>0$,  $\Pot(\x)=L$
for all $\x\ne\0,\1$.
Then a bottleneck argument shows that the mixing time is $\Omega(e^{\beta L})$.
However, for this potential $\Pot$ the maximum global variation $\GV$ coincides with the maximum local variation
$$\LV = \max \{ \Pot(\x) - \Pot(\y) \colon \hamming(\x,\y) = 1 \}$$
and thus it could be compatible with the fact
that the mixing time depends on the maximum local variation and not on $\GV$.
The next theorem proves that the term $e^{\beta \GV}$ cannot be improved for values
of $\beta$ sufficiently large.

\begin{theorem}
\label{th::lb_pot}
For every sequences $\{g_n\}$ and $\{l_n\}$ such that $2\cdot g_n/n\leq  l_n\leq g_n$,
there exists a sequence of potentials $\{\Pot_n\}$ such that
$\Pot_n$ has maximum global variation $\GV_n=g_n$ and
             maximum local variation $\LV_n=l_n$ and
the mixing time $\tm^\beta$ of the Markov chain of the logit dynamics
with inverse noise $\beta$ for the game with potential $\Pot_n$ is
$$\tm^\beta\geq e^{\beta\GV_n(1-o(1))}.$$
\end{theorem}
\begin{proof}
Fix $n$ and set $c=g_n/l_n$.
Consider potential $\Pot_n$ on $\zu^n$ defined as
$$
\Pot_n(\x)=-l_n\cdot\min \left\{ c, \left|c - w(\x) \right| \right\},
$$
where $w(\x)$ denotes the number of $1$'s in $\x$.
Notice that $\LV_n=l_n$, the minimum of $\Pot_n$ is $\Pot_n(\0)=-cl_n=-g_n$,
while the maximum is 0 and is attained by the states in the set
$M=\left\{\x\in\zu^n\colon w(\x)=c\right\}$. Therefore $\GV_n=g_n$.
Moreover, we have that
if $|w(\x)-c|=|w(\y)-c|$ then $\Pot_n(\x)=\Pot_n(\y)$;
that is, $\Pot_n$ is symmetrical with respect to $c$
(note that $c\leq n/2$, by hypothesis).

Consider the Hamming graph on  $S \setminus M$:
it is not hard to see that this graph has two connected components.
Let us name $R$ the component that contains $\0$, i.e,
$$R=\left\{\x\in\zu^n\colon w(\x)<c\right\}.$$
By the symmetry of $\Pot_n$,
we have that $\pi(R)\leq\frac{1}{2}$ and thus we can apply
the Bottleneck Ratio Theorem (Theorem~\ref{theorem:bottleneck}) to $R$.
Let us now provide a bound on $B(R)$.
We start by observing that for every pair $\x,\y$ of profiles
that differ in player $i$, for some $i\in[n]$,
it holds that
\begin{equation}
\label{eq:boundQ(x,y)}
 Q(\x,\y)
= \frac{e^{-\beta \Pot(\x)}}{Z} \cdot \frac{1}{n} \cdot
\frac{e^{-\beta \Pot(\y)}}{\sum_{s \in S_i} e^{-\beta \Pot(\x_{-i}, s)}}
\leq\frac{e^{-\beta \Pot(\y)}}{nZ}\,.
\end{equation}
We define $\partial R$ as the set of profiles in $R$
that have a neighbor in $M$, i.e.,
$$
\partial R = \left\{\x\in\zu^n\colon w(\x)=c-1\right\}.
$$
Note that for every $\x \in \partial R$
there are at most $(m-1)n$ neighbors outside $R$ and all of them belong to $M$ by definition, thus
$$
Q\left(R, \overline{R}\right)
=        \sum_{\x \in \partial R} \sum_{\y \in M}
		Q(\x,\y) \\
\leq \sum_{\x \in \partial R} \sum_{\y \in M}
		\frac{e^{-\beta \Pot(\y)}}{nZ}
\leq (m-1) \left| \partial R \right| \frac{1}{Z}\,,
$$
where the last inequality follows from $\Pot(\y) = 0$ for every $\y \in M$.
Obviously, we have
$$
\pi(R) \geq \pi(\0) = \frac{e^{\beta\GV_n}}{Z}.
$$
The last two inequalities yield
$$
B(R) = \frac{Q\left(R, \overline{R}\right)}{\pi(R)} \leq \frac{(m - 1)\left|\partial R\right|}{e^{\beta\GV_n}}\,.
$$
Finally, observe that
$$
\left|\partial R\right| \leq
    \binom{n}{c} \leq e^{c \log n} = e^{\frac{\GV_n\cdot\log n}{\LV_n}}.
$$
Thus, by Theorem~\ref{theorem:bottleneck},
we have that the mixing time of Markov chain of the logit dynamics with inverse noise $\beta$
for the game with potential $\Pot_n$ is
$$
\tm^\beta(\varepsilon)\geq\frac{1-2\varepsilon}{2(m-1)} \cdot e^{\beta \GV_n - (\GV_n/\LV_n) \log n}=
e^{\beta\GV(1-o(1))}
$$
as $\beta$ can be chosen large enough so that
$\log n/\LV_n=o(1)$.
\end{proof}

\subsection{\texorpdfstring{An upper bound for small $\beta$}{An upper bound for small beta}}
The upper bound provided by~Theorem~\ref{thm:potential:ub}  is very weak  as it
gives a quadratic upper bound even for small $\beta$.
On the other hand, by Lemma~\ref{le:hypercube},
the mixing time of $\M^0$ is $\OO(n\log n)$.
In this section we will show that this bound holds for small values of
$\beta$ and not only for $\beta=0$.
\begin{theorem}
\label{th:mixing_pot_small}
Let $\G$ be an $n$-player potential game with maximum local variation $\LV$.
If $\beta \leq c/(n\cdot\LV )$, for some constant $c < 1$,
then the mixing time of the Markov chain $\MB$ of the logit dynamics
for $\G$
with inverse noise $\beta$ is $\tm^\beta=\OO(n \log n)$.
\end{theorem}
\begin{proof}
Let $P$ be the transition matrix of $\MB$.
For every  pair of profiles $\x$ and $\y$ at Hamming distance $\hamming(\x,\y)=1$,
we describe a coupling $(X,Y)$ of the distributions $P(\x,\cdot)$ and $P(\y,\cdot)$.

For each player $i$, we partition two copies, $I_{X,i}$ and $I_{Y,i}$,
of the interval $[0,1]$ in sub-intervals each labeled with a strategy
from the set $S_i=\{z_1,\ldots,z_{|S_i|}\}$ of strategies of player $i$.
For each $U\in [0, 1]$, we define the {\em label of $U$ with respect to $X$
and $i$} to be the label of the sub-interval of $I_{X,i}$ that contains $U$.
Similarly, we define the {\em label of $U$ with respect to $Y$ and $i$}.
The coupling picks $i\in[n]$ and $U\in[0,1]$ uniformly at random.
$X$ and $Y$ are then updated
by setting the strategy of the $i$-th player
equal to the label of $U$ with respect to $X$ and  $i$
and, similarly, $Y$ is updated
by setting the strategy of the $i$-th player
equal to the label of $U$ with respect to $Y$ and $i$.

Let us now describe how the sub-intervals are constructed.
At the beginning $I_{X,i}$ and $I_{Y,i}$ are each considered as one
unlabeled interval. As the construction of the sub-intervals progresses,
we cut pieces from the unlabeled intervals and label them with a strategy.
Specifically, the partitioning of $I_{X,i}$ and $I_{Y,i}$
has one phase for each $z\in S_i$ during which
the leftmost portion of length $l=\min\{\sigma_i(z\mid\x),\sigma_i(z \mid\y)\}$
of the unlabeled interval of both $I_{X,i}$ and $I_{Y,i}$ is
labeled with $z$.
In addition,
the rightmost portion of length $\sigma_i(z\mid\x)-l$
of the unlabeled interval of $I_{X,i}$
and
the rightmost portion of length $\sigma_i(z\mid\y)-l$
of the unlabeled interval of $I_{Y,i}$
are both labeled with $z$.
Notice that, for each $z$, at least one of these last two sub-intervals has length $0$.
When all $z\in S$ have been considered, all points of the two intervals
belong to a labeled sub-interval.
We have that $(X,Y)$ is a coupling of $P(\x,\cdot)$ and $P(\y,\cdot)$. Indeed, by construction
the size of the sub-intervals labeled $z$ in $I_{X,i}$ is exactly $\sigma_i(z\mid\x)$ and
similarly for $I_{Y,i}$.

Moreover, there exists $\ell_i \in (0,1)$ such that all $U\leq\ell_i$ have
the same label with respect to $X$ and $Y$;
whereas, for all $U>\ell_i$
the label with respect to $X$ and
the label with respect to $Y$ differ.

Let us now bound the expected distance $\hamming(X,Y)$ for
profiles $\x$ and $\y$ differing in the $j$-th player.
We have three possible cases.
\begin{itemize}
\item $i=j$.

In this case, which has probability $1/n$,
we have that $\sigma_j(z\mid\x)=\sigma_j(z\mid\y)$ for all $z\in S_j$ and
thus $\ell_j=1$. As a consequence, for any $U$ the same strategy is
used to update the $j$-th player and we have that $X=Y$ and thus
$\hamming(X,Y)=0$.

\item $i\ne j$ and $U \leqslant \ell_i$.

This case has probability $\frac{1}{n}\sum_{i\ne j}\ell_i$ and
in this case $U$ has the same label with respect to $X$ and $Y$.
Therefore, the $i$-th player strategy is update in the same way,
thus keeping the distance between $X$ and $Y$ equal to $1$.

\item $i\ne j$ and $U>\ell_i$.

This case has probability $\frac{1}{n}\sum_{i\ne j}(1-\ell_i)$ and
in this case $U$ has different labels with respect to $X$ and $Y$.
Therefore, the $i$-th player strategy is update differently in $X$ and $Y$
thus bringing the distance between $X$ and $Y$ up to $2$.

\end{itemize}
We have that the expected distance between $X$ and $Y$ is
\begin{align*}
\Expec{\x,\y}{\hamming(X,Y)}
& =     \frac{1}{n} \sum_{i\ne j} \left(\ell_i + 2 (1 - \ell_i)\right) \\
& \leq  \left(1-\frac{1}{n}\right)(2 - \ell)\,
\end{align*}
where $\ell=\min_i \ell_i$.

Let us now give a lower bound on $\ell$.
From the coupling construction,
we have
$\ell_i = \sum_{z \in S_i} \min \{ \sigma_i(z \mid \x), \sigma_i(z \mid \y) \}$.
Observe that for any profile $\x$, any player $i$ and any strategy $z \in S_i$ it holds that
$$
\sigma_i(z \mid \x) = \frac{e^{-\beta \Pot(\x_{-i}, z)}}{\sum_{k \in S_i}e^{-\beta \Pot(\x_{-i}, k)}}
= \frac{1}{\sum_{k \in S_i}e^{\beta \left[ \Pot(\x_{-i}, z) - \Pot(\x_{-i}, k) \right]}}
\geq \frac{1}{|S_i| e^{\beta \LV}}\,.
$$
Hence
$$
\ell_i = \sum_{z \in S_i} \min \{ \sigma_i(z \mid \x), \, \sigma_i(z \mid \y) \} \geq \sum_{z \in S_i} \frac{1}{|S_i| e^{\beta \LV}} = e^{- \beta \LV} \geq e^{-c/n}\,,
$$
where in the last inequality we used the fact that $\beta<c/(n\cdot\LV)$.
Thus, the expected distance between $X$ and $Y$ is upper bounded by
\begin{align*}
\Expec{\x,\y}{\hamming(X,Y)} & \leq e^{-1/n}(2 - \ell) \leq e^{-1/n}(2 - e^{-c/n}) \\
& = e^{-1/n}(1 + (1 - e^{-c/n})) \leq e^{-1/n}(1 + c/n) \leq e^{-1/n}e^{c/n} = e^{-\frac{1-c}{n}}\,,
\end{align*}
where in the second line we repeatedly used the well-known inequality $1 + x \leqslant e^x$ for every $x > -1$.

The thesis then follows by applying Theorem~\ref{theorem:pathcoupling} with $\alpha = \frac{1-c}{n}$.
\end{proof}

\subsection{\texorpdfstring{Bounds for large $\beta$}{For large beta}}
In this section we give another upper bound on the mixing time of the Markov chain of the logit dynamics with inverse noise $\beta$ for potential games.
The bound depends on a structural property of the potential function that measures
the difficulty of visiting a profile starting from a profile with higher potential value.
We will also show that the upper bound is tight for large values of $\beta$.

Let $\Pot$ be an $n$-player potential defined over the set $S$ of profiles of the game.
We consider paths $\gamma$ over the Hamming graph of $S$. More precisely,
a path $\gamma$ is a sequence $(\x_0, \x_1, \ldots, \x_k)$ of profiles
such that $\x_{i-1}$ and $\x_i$ differ in the strategy of
exactly one player, for $i=1,\ldots,k$.
We denote by $\mathcal{P}_{\x,\y}$ the set of all paths from $\x$ to $\y$.
For a path $\gamma=(\x_0, \x_1, \ldots, \x_k)$ with $\Phi(\x_0)\geq\Phi(\x_k)$,
we denote by $\zeta(\gamma)$
the maximum increase of the potential function along the path; that is,
$$
\zeta(\gamma)=\max_{0 \leq i \leq k} \left(\Pot(\x_i) - \Pot(\x_0)\right),
$$
and set
$$
 \zeta(\x,\y) = \min_{\gamma\in\mathcal{P}_{\x,\y}} \zeta(\gamma)\,,
$$
Finally, we set $\zeta = \max_{\x, \y} \zeta(\x,\y)$. Notice that $\zeta \geq 0$.

\begin{lemma}
\label{pot_rel}
The relaxation time of the Markov chain of the logit dynamics  with inverse noise $\beta$ for a
potential game with potential $\Pot$ and at most $m$ strategies for player satisfies
 $$
  \tr^\beta \leq n m^{2n+1} e^{\beta\zeta}\,.
 $$
\end{lemma}
\begin{proof}
Recall that an edge $(\u,\z)$ of $\MB$ is {\em admissible} if one of $\u$ and $\z$
minimizes the potential in the intersection of the neighborhoods of $\u$ and $\z$.
To bound the relaxation time, we construct,
for each pair $(\x,\y)$ of profiles of $\MB$, a path $\gamma_{\x,\y}$ in $\MB$ such
that $\zeta(\gamma_{\x,\y})=\zeta(\x,\y)$ and each edge of $\gamma_{\x,\y}$ is admissible.
Every pair $(\x,\y)$ of profiles has such a path.
Indeed, let $\gamma$ be a path from $\x$ to $\y$ with $\zeta(\gamma)=\zeta(\x,\y)$ and
let $(\u,\v)$ be a non-admissible edge of $\gamma$. Then we replace
the edge $(\u,\v)$ with the two admissible edges $(\u,\z)$ and $(\z,\v)$, where
$\z$ is the profile that minimizes the potential over the profiles in the intersection
of the neighborhoods of $\u$ and $\v$. Notice that
the resulting path $\gamma'$ has one fewer non-admissible edge
and $\zeta(\gamma')=\zeta(\gamma)$.

We next lower bound $Q(\u,\z)$ for an an admissible edge $(\u,\z)$.
Notice that since the chain is reversible we have
$Q(\u,\z)=Q(\z,\u)$ and thus, without loss of generality, we can
compute $Q$ only for admissible edges $(\u,\z)$ in which $\z$ has the minimum
potential from among the common neighbor profiles of $\u$ and $\z$.
For such an admissible edge $(\u,\z)$,
by denoting by $j$ the player at which $\u$ and $\z$ differ, we have
\begin{align*}
  Q(\u,\z) & = \pi(\u) P(\u,\z) =
        \pi(\u) \cdot \frac{1}{n} \cdot \frac{e^{-\beta \Phi(\z)}}{\sum_{s \in S_j}
            e^{-\beta\Phi(\u_{-j}, s)}}\\
  & = \pi(\u) \cdot \frac{1}{n} \cdot\frac{1}{\sum_{s \in S_j}
            e^{-\beta[\Phi(\u_{-j},s)-\Phi(\z)]}}
  \geq  \frac{\pi(\u)}{nm}\,,
 \end{align*}
where the last inequality follows from the fact that profile $\z$
has the minimum of the potential over profiles $(\u_{-j},s)$, with $s\in S_j$,
and thus all the addends of the sum are at most $1$.

Then, for every admissible edge $(\u,\z)$, we obtain
 \begin{equation}
  \label{eq:zeta_bound}
  \begin{split}
\sum_{\begin{subarray}{c}(\x,\y) \colon\\(\u,\z) \in \gamma_{\x,\y}\end{subarray}}
        \frac{\pi(\x)\pi(\y)}{Q(\u,\z)}\cdot|\gamma_{\x,\y}|
& \leq
  nm^{n+1} \sum_{\begin{subarray}{c}(\x,\y) \colon\\(\u,\z) \in \gamma_{\x,\y}\end{subarray}} \pi(\y) e^{\beta (\Pot(\u) - \Pot(\x))}\\
  & \leq
  nm^{n+1} \cdot e^{\beta\zeta} \sum_{\begin{subarray}{c}(\x,\y) \colon\\(\u,\z) \in \gamma_{\x,\y}\end{subarray}} \pi(\y)\\
  & \leq nm^{2n+1} e^{\beta\zeta}\,,
  \end{split}
 \end{equation}
where we used the fact that $|\gamma_{\x,\y}|\leq m^n$ and
that admissible edge $(\u,\z)$ can be used by  at most $m^{n}$ paths.
The lemma then follows from Lemma~\ref{comp_lemma} and Theorem~\ref{thm:diod_conj}.
\end{proof}
Notice that,
if $\zeta > 0$, then $\tr^\beta \leq e^{\beta\zeta(1 + o(1))}$
(the asymptotic is on $\beta$).
By using this bound on relaxation time and the fact that
$\pi_{\min} \geq 1/\left(e^{\beta\Delta\Pot} |S|\right)$,
Theorem~\ref{theorem:relaxation} gives the following bound on the mixing time.
\begin{theorem}
\label{th:mixing_pot_high}
The mixing time of the Markov chain of the logit dynamics  with inverse noise $\beta$ for a
potential game with potential $\Pot$, $\zeta>0$, and at most $m$ strategies for player satisfies
 $$
  \tm^\beta \leq e^{\beta\zeta(1 + o(1))}\,.
 $$
\end{theorem}

The next theorem shows that this bound is almost tight.
\begin{theorem}
 \label{pot_lower}
The mixing time of the Markov chain of the logit dynamics  with inverse noise $\beta$ for a
potential game with potential $\Pot$, $\zeta>0$, and at most $m$ strategies for player satisfies
 $$
  \tm^\beta \geq e^{\beta\zeta(1 - o(1))}\,.
 $$
\end{theorem}
\begin{proof}
Let $(\x,\y)$ be a pair such that  $\zeta=\zeta(\x,\y)$,
let $\gamma_{\x,\y}$ be a path such that $\zeta(\gamma_{\x,\y}) = \zeta(\x,\y)$
and let $\z$ be the profile of maximum potential along $\gamma_{\x,\y}$.
Therefore, $\zeta = \Pot(\z) - \Pot(\x)$.
Now consider the set $\mathcal{P}_{\x,\y}$ of all paths from $\x$ to $\y$ and
let $M$ be the set of profiles that have maximum potential along at least one path of
$\mathcal{P}_{\x,\y}$.
Since $\zeta > 0$, then $M\neq S$ and $M\ne\emptyset$.
By removing all vertices in $M$, we disconnect $\x$ from $\y$ in the Hamming graph
and we let $R_{\x,M}$ (respectively, $R_{\y,M}$) be the component that contains $\x$
(respectively, $\y$).

Obviously, at least one of $R_{\x,M}$ and $R_{\y,M}$ has
total weight less than $1/2$ in the stationary distribution $\pi$.
We next show that there exists one of $R_{\x,M}$ and $R_{\y,M}$, call it $R$,
such that $\pi(R)<1/2$ and $R$ contains a profile with potential $\Pot(\x)$.
This is obvious when $\Pot(\x)=\Pot(\y)$.
Since $(\x,\y)$ attains $\zeta$ it must be the case that $\Pot(\x)\geq\Pot(\y)$ and
thus we only need to consider the case $\Pot(\x)>\Pot(\y)$.
We claim that in this case $\x$ has the minimum potential in $R_{\x,M}$.
Indeed, suppose $\Pot(\x)>\Pot(\w)$ for some $\w \in R_{\x,M}$ and consider the path
$\gamma_{\w,\y}$ from $\w$ to $\y$ such that $\zeta(\gamma_{\w,\y}) = \zeta(\w,\y)$.
Observe that both endpoints of $\gamma_{\w,\y}$ have potential strictly lower than $\x$.
Moreover, since $M$ disconnect $\w$ and $\y$,
$\gamma_{\z,\y}$ has to go through a profile of $M$ with potential at least $\Pot(\z)$.
Then $\zeta(\gamma_{\z,\y})>\Pot(\z)-\Pot(\x) = \zeta$ and this is a contradiction.
Therefore, if $\Pot(\x) > \Pot(\y)$, then $\Pot(\w) > \Pot(\y)$ for each $\w \in R_{\x,M}$
and hence $\pi(\y) > \pi(R_{\x,M})$, for $\beta$ sufficiently large.
Thus, if $\Pot(\x)>\Pot(\y)$, then $\pi(R_{\x,M}) \leq 1/2$ for $\beta$ sufficiently large.

Next we provide a bound on $B(R)$.
Denote by $\partial R$ the set of profiles in $R$ that have a neighbor outside $R$
(and hence in $M$) and note that each profile $\u \in \partial R$ has at most $n(m-1)$ neighbors in $M$.
We have
$$
Q\left(R, \overline{R}\right) =
        \sum_{\u \in \partial R} \sum_{\v \in M}
		Q(\u,\v) \\
\leq \sum_{\u \in \partial R} \sum_{\v \in M}
		\frac{e^{-\beta \Pot(\v)}}{nZ}
\leq (m-1) \left| \partial R \right| \frac{e^{-\beta\Pot(\z)}}{Z}\,,
$$
where the first inequality follows from \eqref{eq:boundQ(x,y)} and last inequality follows from $\Pot(\v) > \Pot(\z)$ for every $\v \in M$.
Obviously, we have
$$
\pi(R) \geq \pi(\x) = \frac{e^{-\beta\Pot(\x)}}{Z}\,.
$$
The last two inequalities yield
$$
B(R) = \frac{Q\left(R, \overline{R}\right)}{\pi(R)} \leq (m - 1)\left|\partial R\right| \cdot e^{-\beta\zeta}\,.
$$
Thus, by Theorem~\ref{theorem:bottleneck},
we have that the mixing time of Markov chain of the logit dynamics with inverse noise $\beta$
for the game with potential $\Pot_n$ is
$$
\tm^\beta(\varepsilon) = \frac{1-2\varepsilon}{2(m-1)\left|\partial R\right|} \cdot e^{\beta \zeta}
$$
and the theorem follows by taking $\beta$ sufficiently large.
\end{proof}

\section{\texorpdfstring{Mixing time independent of $\beta$}{Mixing time independent of beta}}\label{sec::dominant}
Theorems~\ref{pot_lower} and \ref{th::lb_pot} show that there are games where the mixing time grows with $\beta$.
In this section we show that
if $\G$ has a {dominant profile} then
the mixing time of its logit dynamics
is upper bounded by a constant independent of $\beta$.
We say that strategy $s\in S_i$ is a {\em dominant strategy} for player $i$
if for all $s' \in S_i $ and all strategy profiles $\x \in S$,
$$
 u_i(s, \x_{-i}) \geq u_i(s', \x_{-i})\,.
$$
A {\em dominant profile} $(s_1,\ldots,s_n)$ is a profile in which
$s_i$ is a dominat strategy for player $i=1,\ldots,n$.

Let $\G$ be a game with a dominant profile.
We assume without of loss of generality that all players have the
same dominant strategy and let us name it $0$.
Thus profile $\0=(0,\ldots,0)$ is the {dominant profile}.
It is easy to see that the following holds for the logit dynamics of $\G$.
\begin{obs}\label{obs:all-chosen}
Let $\G$ be an $n$-player game with dominant strategies.
Then for every profile $\x$, for every $\beta \geq 0$ and for every $i=1,\cdots, n$, we
have $\sigma_i(0 \mid \x) \geq 1 / |S_i|$.
\end{obs}
We are now ready to derive an upper bound on the mixing time of the Markov chain
of the logit dynamics for games with dominant strategies.

\begin{theorem}
\label{theorem:dominant-strategies}
Let $\G$ be an $n$-player games with dominant profile where each player has at most $m$ strategies.
The mixing time of $\MB(\G)$ is $\tm^\beta=\OO\left(m^n n\log n\right)$.
\end{theorem}
\begin{proof}
We apply the coupling\index{coupling} technique (see Theorem~\ref{thm:coupling}) and use
the same  coupling used in the proof of Theorem~\ref{th:mixing_pot_small}.
The following properties of that coupling are easily verified.
\begin{enumerate}
\item The coupling always selects the same player in both chains.
\item If player $i$ is selected for update,
the probability that both chains choose strategy $s$ for player $i$
is exactly $\min\{\sigma_i(s\mid\x),\sigma_i(s\mid\y)\}$.
Since, by Observation~\ref{obs:all-chosen},
we have that $\sigma_i(0\mid\x),\sigma_i(0\mid\y)\geq 1/|S_i|$, then
the probability that the coupling updates the strategy of player $i$ in both chains to $0$
is at least $1/|S_i|$.
\item Once the two chains coalesce, they stay together.
\end{enumerate}
Let $\tau$ be the first time such that all the players have been selected at least once and let $t^\star=2 n \log n$. Note that, by a Coupon Collector's argument and by Markov's inequality,
 for all starting profiles $\z$ and $\w$ of the coupled chains,
\begin{equation}
 \label{eq:coupon}
 \Prob{\z,\w}{\tau \leqslant t^\star}\leq \frac{1}{2}\,.
\end{equation}
We next observe that for all starting profiles $\z$ and $\w$, it holds that
\begin{equation}\label{eq:all-zero}
\Prob{\z,\w}{X_{t^\star} = \0 \mbox{ and } Y_{t^\star} = \0 \mid \tau
\leqslant t^\star} \geqslant \frac{1}{m^n}\,.
\end{equation}
Indeed, given that all players have been selected at least once within time $t^\star$, both chains are in profile $\0$ at time $t^\star$ if and only if every player chose strategy $0$ in both chains the last time she played before time $t^\star$.
From the properties of the coupling and by Observation~\ref{obs:all-chosen},
it follows that this event occurs with probability at least $1/m^n$.
Hence, for all starting profiles $\z$ and $\w$, we have that
\begin{equation}
 \label{eq:dominantt}
 \begin{split}
  \Prob{\z,\w}{X_{t^\star} = Y_{t^\star}} & \geq \Prob{\z,\w}{X_{t^\star} = \0 \mbox{ and } Y_{t^\star} = \0}\\
  & \geq \Prob{\z,\w}{X_{t^\star} = \0 \mbox{ and } Y_{t^\star} = \0 \;|\, \tau  \leq t^\star}\Prob{\z,\w}{\tau \leqslant t^\star}\\
  & \geq \frac{1}{m^n} \cdot \frac{1}{2}\,,
 \end{split}
\end{equation}
where in the last inequality we used \eqref{eq:all-zero} and \eqref{eq:coupon}.

Therefore, by considering $k$ phases each one lasting $t^\star$ time steps, since the bound in (\ref{eq:dominantt}) holds for every starting states of the Markov chain, we have that the probability that the two chains have not yet coupled after $k t^\star$ time steps is
$$
\Prob{\x,\y}{X_{kt^\star} \neq Y_{kt^\star}} \leqslant \left( 1 - \frac{1}{2m^n} \right)^k \leqslant
e^{- k/ 2 m^n}\,,
$$
which is less than $1/4$, for $k = \OO(m^n)$. By applying the Coupling Theorem (see Theorem~\ref{thm:coupling}) we have that $\tm = \OO\left(m^n n \log n\right)$.
\end{proof}

\smallskip\noindent{\it Remark.}
We can write the result of the previous theorem by saying that if a game $\G$ has dominant strategies
then the mixing time of its logit dynamics is $O(1)$, where the asymptotic is on $\beta$.

We next prove that, for every $m \geqslant 2$, there are $n$-player games with $m$ strategies per player whose logit dynamics mixing time is $\Omega\left({m^{n-1}}\right)$. Thus the $m^n$ factor in the upper bound given by Theorem~\ref{theorem:dominant-strategies} cannot be essentially improved.

\begin{theorem}
\label{lb_dominant}
For every $m\geq 2$ and $n\geq 2$, there exists a $n$-player potential game with dominant strategies where each player has $m$ strategies and such that, for sufficiently large $\beta$, the mixing time of the Markov chain of the logit dynamics with inverse noise $\beta$ is $\tm=\Omega\left({m^{n-1}}\right)$.
\end{theorem}
\begin{proof}
Consider the game with $n$ players, each of them having strategies $\{0,\ldots,m-1\}$, such that for every player $i$:
$$
u_i(\x) = \begin{cases}
           0,  & \mbox{if } \x=\0; \\
	   -1, & \mbox{otherwise.}
          \end{cases}
$$
Note that $0$ is a dominant strategy. This is a potential game with potential $\Pot(\x)=-u_i(\x)$ and thus the stationary distribution is given by the Gibbs measure in~\eqref{eq:Gibbs}. We apply the bottleneck ratio\index{bottleneck ratio} technique (see Theorem \ref{theorem:bottleneck}) with $R=\{0 \ldots,m-1\}^n \setminus \{\0\}$, for which we have
$$
\pi(R) = \frac{e^{-\beta}}{Z}(m^n-1)
$$
with $Z={1+e^{-\beta}(m^n-1)}$. It is easy to see that $\pi(R)<1/2$ for $\beta > \log(m^n-1)$ and furthermore
\begin{align*}
Q(R,\overline R) & = \sum_{\x\in R}\pi(\x)P(\x,\0) \\
& = \frac{e^{-\beta}}{Z} \sum_{\x \in R} P(\x, \0) = \frac{e^{-\beta}}{Z} \sum_{\x \in R_1} P(\x, \0)\,,
\end{align*}
where $R_1$ is the subset of $R$ containing all states with exactly one non-zero entry.
Notice that, for every $\x\in R_1$, we have
$$
P(\x,\0)= \frac{1}{n}\cdot \frac{1}{1 + (m-1)e^{-\beta}}\,.
$$
As $|R_1|=n(m-1)$, we have
\begin{align*}
Q(R,\overline R) & = \frac{e^{-\beta}}{Z}|R_1| \frac{1}{n}\cdot \frac{1}{1 + (m-1)e^{-\beta}} \\
& = \frac{e^{-\beta}}{Z} \cdot \frac{m-1}{1 + (m-1)e^{-\beta}}
\end{align*}
whence
\begin{align*}
\tm & \geq \frac{1}{4}\cdot\frac{\pi(R)}{Q(R,\overline R)} \\
& \geq \frac{1}{4}\cdot (m^n-1)\cdot \frac{1+(m-1)e^{-\beta}}{m-1}
> \frac{1}{4} \cdot\frac{m^n - 1}{m-1}\,.\qedhere
\end{align*}
\end{proof}

\paragraph{Max-solvable games.} We note that,
by using the same techniques exploited in this section,
it is possible to prove an upper bound independent of $\beta$ for \emph{max-solvable} games\index{max-solvable game} \cite{NSZ08}, a class which contains games with dominant strategies as a special case, albeit with a function that is much larger than $\OO(m^n n\log n)$.
We omit further details.

\section{Graphical coordination games}\label{sec:coord}
In this section we consider $2\times 2$ \emph{coordination games} played
on a {\em social graph}.
In a coordination game each player has two strategies, $0$ and $1$,
and the payoff are such that the players
have an advantage in selecting the same strategy.
These games are often used to model the spread of
a new technology~\cite{youngTR00}:
the two strategies correspond to adopting or not adopting a new technology.
The game is formally described by the following  payoff matrix
\newsavebox\cgbox
\begin{lrbox}{\cgbox}
\begin{game}{2}{2}
 & $0$ & $1$ \\
$0$ & $a,a$ & $c,d$ \\
$1$ & $d,c$ & $b,b$
\end{game}
\end{lrbox}
\begin{equation}
\label{eq:coorddef}
\usebox{\cgbox}
\end{equation}
For convenience, we set
\begin{equation}
\label{eq:Delta_def}
\dmax := a-d \quad \mbox{ and } \quad \dmin :=b-c\,.
\end{equation}
We assume that $\dmax > 0$ and $\dmin > 0$ which implies that each
player has an incentive in selecting the same strategy as the other
player.
The game has thus two pure Nash equilibria: $(0, 0)$ and $(1, 1)$.
If $\dmax > \dmin$ we say that $(0,0)$ is the
\emph{risk dominant equilibrium};
if $\dmax < \dmin$ we say that $(1, 1)$ is the risk dominant equilibrium;
otherwise, we say that the game has no risk dominant equilibrium.
The risk dominant equilibrium concept is a refinement of the concept
of Nash equilibrium proposed by Harsanyi and Selten \cite{hsMIT88}.
It is easy to see that
a coordination game is a potential game with potential function
$\phi$ defined by
$$\phi(0,0)=-\delta_0,\quad
  \phi(1,1)=-\delta_1,\quad
  \phi(0,1)=\phi(1,0)=0.$$
%In the spread of technology example, we can assume that every player, if unknown about the other player action, prefers to choose the new technology: this assumption will make the outcome where both select the new technology the risk dominant equilibrium.

A \emph{graphical coordination game} is a game in which
$n$ players are connected by a {\em social graph} $G = (V,E)$
and every player plays an instance of the basic coordination game
described by (\ref{eq:coorddef}) with each of her neighbors.
Specifically, a player selects one strategy that is played against
each of her neighbors and the payoff is the sum of the payoffs
of the instances of the basic coordination game played.
It is easy to see that the profiles where all players play the same
strategy are Nash equilibria;
moreover, if $(x,x)$ is risk dominant for the basic coordination game,
then profile $(x, \ldots, x)$ is risk dominant profile for
the graphical coordination game.

We note that graphical coordination games are potential games where the
potential function $\Pot$ is the sum of the edge potentials $\Pot_e$.
For an edge $e=(u,v)$, the edge potential $\Pot_e$ is
defined as $\Pot_e(\x)=\phi(x_u,x_v)$ for edge $e=(u,v)$.
%\begin{equation}
 %\label{eq:pot_coordgames}
 %\Pot_e(\x) = \begin{cases}
		%-\dmax & \text{if } x_u = x_v = 0\,;\\
		%-\dmin & \text{if } x_u = x_v = 1\,;\\
		%0 & \text{otherwise}\,;
              %\end{cases}
%\end{equation}
%with $\dmax$ and $\dmin$ as defined in~(\ref{eq:Delta_def}).

In Section~\ref{subsec:graph} we will show a upper bound on the
mixing time of the logit dynamics with inverse noise $\beta$
for graphical coordination games that holds for every social graph.
Our proof is a generalization of the
the work of Berger et al. \cite{bkmp2005} on the Ising model.
Then, we focus on two of the most studied topologies:
the clique (Section~\ref{subsec:clique2}),
where the mixing time dependence on $e^{\beta\Df(1+o(1))}$ showed in Theorem~\ref{thm:potential:ub} cannot be improved, and the ring (Section~\ref{subsec:ring2}), where a more local interaction implies a faster convergence to the stationary distribution.

In the rest of this section we will assume without loss of generality that $\dmax \geq \dmin$.

\subsection{Graphical coordination games on arbitrary graphs}
\label{subsec:graph}
In this subsection we give a general bound to the mixing time of Markov chain of the logit dynamics with inverse noise $\beta$ for graphical coordination games that depends on the
{\em cutwidth} of the social graph $G$.
For an ordering $\ord$ of the vertices of a graph $G$, we write
``$j <_{\ord} h$'' to denote that vertex $j$ precedes vertex $h$
in the ordering $\ord$ and define, for every vertex $i$,
\begin{equation}
\label{eq:crossedges}
E^{\ord}_i = \{ (j, h) \in E | \, j \leq_{\ord} i <_{\ord} h \}\,.
\end{equation}
Finally,
the \emph{cutwidth} $\chi(\ord)$ of ordering
$\ord$ is $\chi(\ord)=\max_i |E^{\ord}_i|$ and the
\emph{cutwidth} $\chi(G)$ of $G$ is
\begin{equation}
\label{eq:cutwidth}
 \chi(G) = \min_{\ord \in \ordset} \chi(\ord),
\end{equation}
where $\ordset$ is the set of all the orderings on the vertices in $V$.

The following theorem is the main result of this subsection.
\begin{theorem}
 \label{th:all_graph}
Let $\G$ be an $n$-player graphical coordination game on social graph $G$.
The mixing time of the Markov chain of the logit dynamics with inverse noise $\beta$ for $\G$ is
 $$
  \tm^\beta \leq 2n^3 e^{\chi(G)(\dmax + \dmin)\beta} (n\dmax \beta +1)\,.
 $$
\end{theorem}
The proof is based on the Canonical Paths method (see Theorem \ref{comp_lemma}).
Fix an ordering $\ord$ of the vertices $\{1,2,\ldots,n\}$
of $G$ and rename the
vertices so that $1<_\ord 2<_\ord\cdots <_\ord n$.
We next define a set of paths $\Gamma^\ord=\{\Gamma_{\x,\y}^\ord\}$,
one for every pair $(\x,\y)$ of profiles, and
relate the congestion of $\Gamma^\ord$ to the cutwidth of $\ord$.
Specifically, the path $\Gamma_{\x,\y}^\ord$ from $\x$ to $\y$ relative to ordering $\ord$
is defined as follows. Denote by $i_1\ldots,i_d$,
with $0<d\leq n$, the components in which
$\x$ and $\y$ differ ordered according to $\ord$.
That is, for $j=1,\ldots,d$, we have $x_{i_j}\ne y_{i_j}$ and
    $i_1<_{\ord} i_{2}<_{\ord}\ldots<_{\ord} i_d.$
Then we set
$\Gamma^{\ord}_{\x, \y}=(\z^0,\z^1,\cdots,\z^d)$,
where $\z^0 = \x, \z^d = \y$ and, for $0<j<d$,
$$\z^j=(y_1,\cdots,y_{i_j},x_{i_j+1},\cdots,x_n).$$
Notice that, since $\x$ and $\y$ coincide in all components from $i_{j}+1$ to $i_{j+1}-1$,
we can also write $\z^j$ as
$$\z^j=(y_1,\cdots,y_{i_{j+1}-1},x_{i_{j+1}},\cdots,x_n).$$
We will also find convenient to define,
for each edge $(\u,\v)$ of $\Gamma^\ord_{\x,\y}$,
$f^{\ord}_{(\u,\v)}(\x,\y)$ as follows
\begin{equation}
 \label{eq:path_edge}
 f^{\ord}_{(\u,\v)}(\x, \y) = \begin{cases}
                (x_1,\cdots,x_{i-1},y_{i},\cdots,y_n), & \text{if }\pi(\u)\leq \pi(\v) \,;\\
                (x_1,\cdots,x_{i},y_{i+1},\cdots,y_n), & \text{if }\pi(\u)>\pi(\v)\,,\\
              \end{cases}
\end{equation}
where $i$ denotes the component in which $\u$ and $\v$ differ.

\smallskip
The following two technical lemmas establish two useful properties of the paths in $\pset^\ord$ and
of the function $f_{\e}$ that will be useful to prove Theorem \ref{th:all_graph}.

\begin{lemma}\label{lemma:canpath}
For every pair of profiles $\x, \y \in \vset$,
for every ordering $\ord$ of the vertices of $G$ and
for every edge $(\u,\v)\in\Gamma^{\ord}_{\x,\y}$,
\begin{equation}
\frac{\pi(\x) \cdot \pi(\y)}{\min\{\pi(\u), \pi(\v)\}} \leq e^{\beta |E^\ord_i|(\dmin + \dmax)} \pi(\z)\,,
\end{equation}
where $\z=f^\ord_{(\u,\v)}(\x,\y)$ and $i$ is the component in which profiles $\u$ and $\v$ differ.
\end{lemma}
\begin{proof}
We only prove the lemma for the case $\pi(\u)\leq\pi(\v)$, the proof for the other case being similar.
Since $(\u,\v)$ is an edge of $\Gamma^\ord_{\x,\y}$ then by definition of the path we have
that $\u=(y_1,\cdots,y_{i-1},x_{i},\cdots,x_n)$. Also, by the definition of the function $f_{(\u,\v)}^\ord$ we have
$\z=(x_1,\cdots,x_{i-1},$ $y_{i},\cdots,y_n)$.
Now observe that
\begin{equation}
 \label{eq:edgepotential}
 \begin{split}
\frac{\pi(\x) \cdot \pi(\y)}{\min\{\pi(\u), \pi(\v)\} \cdot \pi(\z)}
	& = \frac{\pi(\x) \cdot \pi(\y)}{\pi(\u) \cdot \pi(\z)}\\
 	& = e^{-\beta\left( \Pot(\x) + \Pot(\y) - \Pot(\u) - \Pot(\z) \right)}\\
 	& = e^{-\beta \sum_{e \in G}\left( \Pot_e(\x) + \Pot_e(\y) - \Pot_e(\u) - \Pot_e(\z) \right)}\,.
 \end{split}
\end{equation}
Notice that for every edge $e=(a,b)$ of the social graph $G$,
the potential $\Pot_e$ depends only on components $a$ and $b$ of the profile.
It is thus sufficient to prove that, for every edge $e=(a,b)$ of the social graph $G$,
$$
\Pot_e(u_a,u_b)+\Pot_e(z_a,z_b)-\Pot_e(x_a,x_b)-\Pot_e(y_a,y_b)
\begin{cases}
=0,              & {\text{if }} e\notin E_{i};\cr
\leq\dmin+\dmax,    & {\text{if }} e\in E_{i}.
\end{cases}
$$
If $e\notin E_{i}$ then either $a,b\leq_\ord i$ or $a,b>_\ord i$.
In the first case, we have that  $u_a=y_a$ and $u_b=y_b$ and $z_a=x_a$ and $z_b=x_b$. Hence
$$\Pot_e(x_a,x_b) = \Pot_e(z_a,z_b) \qquad\text{and}\qquad \Pot_e(y_a,y_b) = \Pot_e(u_a,u_b)\,. $$
In the second case, we have that $u_a=x_a$ and $u_b=x_b$ and $z_a=y_a$ and $z_b=y_b$. Hence
$$\Pot_e(x_a,x_b) = \Pot_e(u_a,u_b) \qquad\text{and}\qquad \Pot_e(y_a,y_b) = \Pot_e(z_a,z_b)\,. $$
Let us consider the case $e\in E_i$ and assume without loss of generality that $a<_\ord i <_\ord b$.
Therefore we have
$$z_a=x_a, z_b=y_b\qquad\text{and}\qquad u_a=y_a, u_b=x_b$$
and thus we consider quantity
$$M=\Pot_e(y_a,x_b)+\Pot_e(x_a,y_b)-\Pot_e(x_a,x_b)-\Pot_e(y_a,y_b)$$
and prove that $M\leq\dmin+\dmax$.
If both addends are equal to $-\dmax$ then we have
$x_a=x_b=0$ and  $y_a=y_b=0$ and thus $M=0$.
In the remaining cases,
at most one of the two negative addends is equal to $-\dmax$ and the other
is either $0$ of $-\dmin$ and thus the claim holds.
\end{proof}

\begin{lemma}\label{lemma:iniective}
For every edge $\e=(\u,\v)$  and for every ordering $\ord$
of the vertices,
        the function $f_\e^\ord$ is injective.
\end{lemma}
\begin{proof}
Let $\z=f_\e^\ord(\x,\y)$ be a profile in the co-domain of $f_\e^\ord$.
We show that $\x$ and $\y$ are uniquely determined by $\z$ and $\e$.
Assume, without loss of generality, that $\pi(\u)\leq\pi(\v)$ and
suppose that $\u$ and $\v$ differ in the $i$-th component.
Then  $(x_1,\ldots,x_{i-1})=(z_1,\ldots,z_{i-1})$ and
      $(y_i,\ldots,y_n)=(z_i,\ldots,z_n)$.
Moreover, since $\e$ belongs to the path $\Gamma^\ell_{\x,\y}$ then
$\v=(y_1,\ldots,y_{i-1},x_i,\ldots,x_n)$ and thus
the remaining components of $\x$ and $\y$ are also determined.
\end{proof}
The following lemma gives an upper bound to the congestion $\rho(\Gamma^\ord)$ of the
set of paths $\Gamma^\ord$ relative to order $\ord$.
\begin{lemma}\label{lemma:congestion}
\begin{equation}
\rho(\pset^{\ord}) \leq 2n^2 e^{\chi(\ord)(\dmax + \dmin)\beta}.
\end{equation}
\end{lemma}
\begin{proof}
Since
$$
P(\u, \v) = \frac{1}{n} \frac{e^{-\beta \phi(\v)}}{e^{-\beta \phi(\v)} + e^{-\beta \phi(\u)}} \geq \frac{e^{-\beta \min\{\phi(\u), \phi(\v)\}}}{2n}
$$
we can state that
\begin{equation}
\label{eq:reversible}
Q(\u, \v) \geq \frac{\min\{\pi(\u), \pi(\v)\}}{2n}.
\end{equation}
Thus,
\begin{align*}
\rho(\pset^{\ord}) & =
\max_{\M-\text{edge } (\u, \v)}
    \left(
        \frac{1}{Q(\u, \v)}
        \sum_{\begin{subarray}{c}\x,\y\colon\\ (\u,\v)\in\Gamma^{\star}_{\x, \y}\end{subarray}}
                 \pi(\x)\cdot\pi(\y)\cdot|\Gamma^{\ord}_{\x,\y}|
     \right) \\
\text{(by Eq.\ref{eq:reversible} and } |\Gamma^\ord_{\x,\y}|\leq n\text{)}\qquad\qquad\qquad	& \leq \max_{\M-\text{edge } (\u, \v)} \left(2n^2 \sum_{\begin{subarray}{c}\x,\y \colon\\ (\u, \v) \in \Gamma^{\ord}_{\x, \y}\end{subarray}} \frac{\pi(\x) \cdot \pi(\y)}{\min\{\pi(\u), \pi(\v)\}}\right) \\
\text{(by Lemma~\ref{lemma:canpath} and } |E_i^\ord|\leq\chi(\ell) \forall i \text{)}\qquad\,	& \leq \max_{\M-\text{edge } \e} \left(2n^2 \sum_{\begin{subarray}{c}\x,\y \colon\\ \e \in \Gamma^{\ord}_{\x, \y}\end{subarray}}  e^{\beta (\dmax + \dmin) \chi(\ell)}
                        \pi\left(f_\e^\ord(\x, \y)\right) \right) \\
	& \leq \max_{\M-\text{edge } \e} \left(2n^2 e^{\beta (\dmax + \dmin) \chi(\ord)} \sum_{\begin{subarray}{c}\x,\y \colon \e \in \Gamma^{\ord}_{\x,\y}\end{subarray}}  \pi(f_\e^\ord(\x,\y))\right) \\
\text{(by Lemma~\ref{lemma:iniective})}\qquad\qquad\qquad\qquad\qquad\,	& \leq 2n^2 e^{\beta \chi(\ord)(\dmax + \dmin)}.\qedhere
\end{align*}
\end{proof}

\begin{proof}[Proof of Theorem \ref{th:all_graph}]
 By Lemma \ref{lemma:congestion} we have that there exists a set of paths $\Gamma$ with
$$
\rho(\Gamma) \leq 2n^2 e^{\chi(G)(\dmax + \dmin)\beta}
$$
and by Theorem~\ref{comp_lemma} and Theorem~\ref{thm:diod_conj},
we obtain that the relaxation time of $\MB$
$$
 \tr^\beta \leq 2n^2 e^{\chi(G)(\dmax + \dmin)\beta}\,.
$$
The theorem then follows by Theorem~\ref{theorem:relaxation} and by observing that
\[
\log \pi_{\min}^{-1} \leq \beta\Df + \log|S| \leq \beta |E|\dmax + \log 2^n \leq n(\beta n \dmax +1)\,.\qedhere
\]
\end{proof}

\subsection{Graphical coordination games on a clique}
\label{subsec:clique2}
We now focus on graphical coordination games played on a \emph{clique}, a very popular network topology where every player plays the basic coordination game in~(\ref{eq:coorddef}) with every other player.
Since the cutwidth of a clique is $\Theta(n^2)$, Theorem~\ref{th:all_graph} gives us an upper bound to the mixing time of the Markov chain of the logit dynamics with inverse noise $\beta$ for this class of games that is exponential in $n^2(\dmax + \dmin)\beta$.

However, since the game is a potential game we can obtain almost tight bounds by using Theorems~\ref{th:mixing_pot_high} and \ref{pot_lower}. Indeed, it is easy to see that the potential $\Pot$ depends only on the number of players playing strategy $1$ and not on their position on the social network: if there are $k$ players playing $1$, then there will be $k(k-1)/2$ edges in the network in which both endpoints are playing $1$, each one contributing $-\dmin$ to the potential, and $(n-k)(n-k-1)/2$ edges whose endpoints are playing $0$ and hence contributing $-\dmax$. Hence, for a profile $\x$ with $k$ players playing $0$, the potential $\Phi(\x) = - \left(\frac{(n-k) (n-k-1)}{2}\dmax + \frac{k (k-1)}{2} \dmin\right)$.

By simple algebraic analysis, it turns out that the maximum of the potential $\Pot_{\max}$ is attained when $k^*$ players are playing strategy $1$, where $k^*$ is the integer closest to $\left\lfloor (n-1) \frac{\dmax}{\dmax + \dmin} + \frac{1}{2}\right\rceil$,  and $\Pot(\x)$ monotonically increases as $w(\x)$ goes from $n$ to $k^\star$ and then monotonically decreases as $w(\x)$ goes from $k^\star$ to $0$. Hence, it turns out that for every pair of profiles $\x,\y$ with $\pi(\x) \leq \pi(\y)$ and any path $\gamma$ between $\x$ and $\y$, the maximum increase of the potential function along $\gamma$ is
\begin{equation}
 \label{eq:zeta_gamma_coord}
 \zeta(\gamma) \geq \begin{cases}
                 \Pot_{\max}-\Pot(\x) & \text{if $w(\x) \leq k^\star$ and $w(\y) > k^\star$ or viceversa;}\\
                 \Pot(\y)   -\Pot(\x) & \text{otherwise.}
                \end{cases}
\end{equation}
Moreover, for any pair of profiles there exist at least a path (e.g., the bit-fixing path) for which \eqref{eq:zeta_gamma_coord} holds with equality. Thus, for the graphical coordination game on a clique, since $\dmax \geq \dmin$, we have $\Phi(\0) \leq \Phi(1)$ and hence $\zeta = \Pot_{\max} - \Pot(\1)$.
Then, by Theorems~\ref{th:mixing_pot_high} and \ref{pot_lower} we can obtain the following  bound to the mixing time of the Markov chain of the logit dynamics with inverse noise $\beta$.
\begin{theorem}
For every graphical coordination game on a clique the mixing time of the Markov chain of the logit dynamics with inverse noise $\beta$ is
$$
C^{\beta (\Pot_{\max} - \Pot(\1))}\leq\tm\leq D^{\beta (\Pot_{\max} - \Pot(\1))\dmin}\,,
$$
where $C,D=\OO_{\beta}(1)$.
\end{theorem}
Note that $\Pot_{\max} - \Pot(\1)$ depends on the ratio $\frac{\dmax}{\dmax + \dmin}$, with the worst case being when no risk dominant strategy exists (i.e., $\dmax = \dmin$): in this case $\frac{\dmax}{\dmax + \dmin} = \frac 1 2$ and $\Pot_{\max} - \Pot(\1) = \Theta(n^2 \dmin)$.

\subsection{Graphical coordination games on a ring}
\label{subsec:ring2}
In this section we give upper and lower bounds on the mixing time
for graphical coordination games played on a \emph{ring}
for basic coordination games
that admit no risk dominant strategy ($\dmax = \dmin = \deq$).
Unlike the clique, the ring encodes a very local type of interaction between the players which is more likely to occur in a social context.
We show, for every $\beta$, an upper bound that is
exponential in $2\deq\beta$
and polynomial in $n \log n$. This improves on the
bounds that can be derived by applying
Theorem~\ref{th:all_graph}  or
Theorem~\ref{th:mixing_pot_high}.
Indeed,
since the cutwidth of the ring is $2$,
Theorem~\ref{th:all_graph}
would give an upper bound exponential in
$4\deq \beta$ and polynomial in $n^4\deq \beta$.
Also, Theorem~\ref{th:mixing_pot_high} only holds for
sufficiently large $\beta$.

The proof of the upper bound uses the path coupling technique (see Theorem~\ref{theorem:pathcoupling}) and can be seen as a generalization of the upper bound on the mixing time for the Ising model on the ring (see Theorem~15.4 in \cite{lpwAMS08}).
\begin{theorem}\label{thm:ubring2}
Let $\G$ be an $n$-player graphical coordination game played on a ring with no risk-dominant strategy ($\dmax = \dmin = \deq$). The mixing time of the Markov chain of the logit dynamics with inverse noise $\beta$ for $\G$ is
 $$
  \tm^\beta = \OO\left(e^{2\deq \beta} n\log n\right).
 $$
\end{theorem}
\begin{proof}
We identify the $n$ players with the integers in $\{0,\ldots,n-1\}$ and assume that every player $i$ plays the basic coordination game with her two adjacent players, $(i-1)\mod n$ and $(i+1)\mod n$.
The set of profile is $\vset=\{0,1\}^n$.

Let us consider two adjacent profiles $\x$ and $\y$
differing in the $j$-th component
and, without loss of generality,
assume that $x_j=1$ and $y_j=0$.
We consider the following coupling for two chains $X$ and $Y$ starting respectively from $X_0=\x$ and $Y_0=\y$: pick $i \in \{0,\ldots,n-1\}$ and $U\in[0,1]$ independently and uniformly at random and update position $i$ of $\x$ and $\y$ by setting
$$
x_i =
\begin{cases}
0,& \text{ if } U\leq \sigma_i(0\mid\x);\\
1,& \text{ if } U>    \sigma_i(0\mid\x);
\end{cases}
\qquad
y_i =
\begin{cases}
0,& \text{ if } U\leq \sigma_i(0\mid\y);\\
1,& \text{ if } U>    \sigma_i(0\mid\y).
\end{cases}
$$
We next compute the expected distance between $X_1$ and $Y_1$,
the states of the two chains after one step of the coupling.
Notice that  $\sigma_i(0 \mid \x)$ only depends on $x_{i-1}$ and $x_{i+1}$ and $\sigma_i(0\mid\y)$ only depends on $y_{i-1}$ and $y_{i+1}$. Therefore, since $\x$ and $\y$ only differ at position $j$, $\sigma_i(0\mid\x)=\sigma_i(0\mid\y)$ for
$i\ne j-1,j+1$.

We start by observing that if position $j$ is chosen for update (this happens with probability $1/n$) then, by the observation above, both chains perform the same update. Since $\x$ and $\y$ differ only for player $j$, we have that the two chains are coupled (and thus at distance $0$). Similarly, if $i \ne j-1,j,j+1$ (which happens with probability $(n-3)/n$) we have that both chains perform the same update and thus remain at distance $1$. Finally, let us consider the case in which $i\in\{j-1,j+1\}$. In this case, since $x_j=1$ and $y_j=0$, we have that $\sigma_i(0 \mid \x)\leq \sigma_i(0 \mid \y)$. Therefore, with probability $\sigma_i(0\mid\x)$ both chains update position $i$
to $0$ and thus remain at distance $1$; with probability $1-\sigma_i(0\mid\y)$ both chains update position $i$ to $1$ and thus remain at distance $1$; and with probability $\sigma_i(0\mid\y)-\sigma_i(0\mid\x)$ chain $X$ updates position $i$ to $1$ and chain $Y$ updates position $i$ to $0$ and thus the two chains go to distance $2$. By summing up, we have that the expected distance $E[\rho(X_1,Y_1)]$ after one step of coupling of the two chains is
\begin{align*}
E[\rho(X_1,Y_1)] & = \frac{n-3}{n} + \frac{1}{n}\sum_{i\in\{j-1,j+1\}} \left[ \sigma_i(0\mid\x)+1-\sigma_i(0\mid\y) + 2\cdot(\sigma_i(0\mid\y)-\sigma_i(0\mid\x))\right] \\
& =
\frac{n-3}{n}+
\frac{1}{n}\cdot \sum_{i\in\{j-1,j+1\}}
(1+\sigma_i(0\mid\y)-\sigma_i(0\mid\x))\\
& =
\frac{n-1}{n}+
\frac{1}{n}\cdot \sum_{i\in\{j-1,j+1\}} (\sigma_i(0\mid\y)-\sigma_i(0\mid\x))\,.
\end{align*}
Let us now evaluate the difference $\sigma_i(0\mid\y)-\sigma_i(0\mid\x)$ for $i=j-1$ (the same computation holds for $i=j+1$). We distinguish two cases depending on the strategies of player $j-2$ and start with the case $x_{j-2}=y_{j-2}=1$. In this case we have that
$$
\sigma_{j-1}(0\mid\x)=\frac{1}{1+e^{2\deq \beta}}
\quad \text{ and } \quad
\sigma_{j-1}(0\mid\y)=\frac{1}{2}\,.
$$
Thus,
$$
\sigma_{j-1}(0\mid\y)-\sigma_{j-1}(0\mid\x)=\frac{1}{2}-\frac{1}{1+e^{2\deq \beta}}.
$$
If instead $x_{j-2}=y_{j-2}=0$, we have
$$
\sigma_{j-1}(0\mid\x)=\frac{1}{2}
\quad \text{ and } \quad
\sigma_{j-1}(0\mid\y)=\frac{1}{1+e^{-2\deq \beta}}\,.
$$
Thus
$$
\sigma_{j-1}(0\mid\y)-\sigma_{j-1}(0|\x) = \frac{1}{1+e^{-2\deq \beta}}-\frac{1}{2} \phantom{very big space}
$$
$$
= 1-\frac{1}{1+e^{2\dmax \beta}}-\frac{1}{2} = \frac{1}{2}-\frac{1}{1+e^{2\deq \beta}}\,.
$$
We can conclude that the expected distance after one step of the chain is
\begin{align*}
E[\rho(X_1,Y_1)] & =
\frac{n-1}{n}+\frac{1}{n}\left(1-\frac{2}{1+e^{2\deq \beta}}\right)\\
& = 1-\frac{2}{n(1+e^{2\deq \beta})} \\
& \leq e^{-\frac{2}{n(1+e^{2\deq \beta})}}.
\end{align*}
Since the diameter of $G$ is $\text{diam}(G)=n$, by applying Theorem~\ref{theorem:pathcoupling} with $\alpha={\frac{2}{n(1+e^{2\deq \beta})}}$, we obtain the theorem.
\end{proof}

The upper bound in Theorem~\ref{thm:ubring2} is nearly tight (up to the
$n\log n$ factor). Indeed, a lower bound can be obtained by applying the Bottleneck Ratio technique (see Theorem~\ref{theorem:bottleneck}) to the set $R=\{\1\}$.
Notice that $\pi(R) \leq \frac{1}{2}$ since profile $\0$ has the same potential
as $\1$. Thus set $R$ satisfies the hypothesis of Theorem~\ref{theorem:bottleneck}.
Simple computations show that the bottleneck ratio is
$$
B(R) = \sum_{\y \neq \1} P(\1,\y) = \frac{1}{1+e^{2\deq \beta}}\,.
$$
Thus, by applying Theorem~\ref{theorem:bottleneck}, we obtain the following bound.
\begin{theorem}\label{thm:lbring2}
Let $\G$ be a $n$-player graphical coordination game on a ring with no risk-dominant strategy. The mixing time of the logit dynamics for $\G$ is $\tm = \Omega\left(1 + e^{2\deq \beta}\right)$.
\end{theorem}

\section{Conclusions and open problems}\label{sec::conclusions_spaa}
In this chapter we give different bounds on the mixing time of the logit dynamics for the class of potential games: we showed that the mixing time is fast when $\beta$ is small enough and we found the structural property of the game that characterizes the mixing time for large $\beta$; finally we showed a bound that holds for every value of $\beta$. Unfortunately, the last bound does not match the previous ones in every game: thus it would be interesting to have an unique bound that holds for every values of $\beta$ and matches with bounds given in Theorems~\ref{th:mixing_pot_small} and \ref{th:mixing_pot_high}.

On the other hand, we show that there exists a class of games, namely dominant strategy games, such that the mixing time of the logit dynamics does not grow indefinitely with the inverse noise. However, it would be interesting to characterize the class of games for which this holds: a natural candidate for this class may be the set of potential games where $\zeta^\star = 0$.

Finally, we consider coordination games on different topologies of graphs, where we give evidence that the mixing time is affected by structural properties as the connectedness of the network.

\smallskip
The main goal of this line of research is to give general bounds on the logit dynamics mixing time for any game, highlighting the features of the game that distinguish between polynomial and exponential mixing time. We stress that, when the game is not a potential game, in general there is not a simple closed form for the stationary distribution like Equation~(\ref{eq:Gibbs}).

At every step of the logit dynamics one single player is selected to update her strategy. It would be interesting to consider variations of such dynamics where players are allowed to update their strategies simultaneously. The special case of parallel best response (that is $\beta=\infty$) has been studied in~\cite{NSZ08}. Another interesting variant of the logit dynamics is the one in which the value of $\beta$ is not fixed, but varies according to some {\em learning process} by which players acquire more information on the game as time progresses.

When the mixing time of the logit dynamics is polynomial,
we know that the stationary distribution gives good predictions of the state of the system after a polynomial number of time steps. When the mixing time is exponential, it would be interesting to analyze the \emph{transient} phase of the logit dynamics, in order to investigate what kind of predictions can be made about the state of the system in such a phase. Preliminary results along this direction has been achieved in \cite{afppSODA12}.

\bibliographystyle{plain}
\bibliography{logit}

\end{document}